%% file: main.tex
\documentclass[11pt]{article}
\usepackage{amsmath, amssymb, amsthm} 

\usepackage{amsmath,graphicx,tikz}
\usetikzlibrary{arrows.meta,decorations.pathreplacing}

\input{macros}

\begin{document}

\title{Lower Bounds for Parallel Diffusion Sampling}
\author{Yiwen Kou\thanks{Email: evankou@ucla.edu, supported by NSF EnCORE: Institute for Emerging CORE Methods in Data Science Award \#2217033 and NSF AF: Small Award \#2425350} \\ UCLA \and  Yimeng Wang\thanks{Email: kobewang1359@ucla.edu, supported by NSF EnCORE: Institute for Emerging CORE Methods in Data Science Award \#2217033 and NSF AF: Small Award \#2425350 } \\ UCLA}

\maketitle

\date{\today}

\begin{abstract}
\normalsize
   Standard diffusion samplers generate samples through sequential evaluations of a learned score function, with each query depending on previous answers. Parallel samplers seek to accelerate this process by making multiple score queries simultaneously. We ask how many adaptive rounds are necessary when each round allows polynomially many queries.

   We establish the first polynomial lower bounds on the round complexity of diffusion sampling with approximate scores satisfying a mean-square accuracy guarantee. Specifically, we prove (1) a $\widetilde{\Omega}(d^{1/3})$-round lower bound for sampling smooth, near-isotropic Gaussian mixtures in $\rr^d$, and (2) a tight $\Omega(d)$-round lower bound for uniform sampling from anisotropic axis-aligned boxes contained in the unit ball. Both bounds hold for arbitrary randomized algorithms making polynomially many queries per round at arbitrary locations and noise levels, with inverse-polynomial score error and constant total variation accuracy. 
   Our near-isotropic construction follows the blockwise approach of earlier lower bounds for parallel matroid search and discrete sampling \cite{KUW85,AGR24}. We use independent coordinate blocks containing random codebooks with well-separated rates and construct approximate score oracles whose responses in each block depend on consistency tests in earlier blocks. Together, the rate separation and these tests ensure that each round reveals information about at most one additional block with high probability.
\end{abstract}
\newpage


\section{Introduction}\label{sec:introduction}
\input{introduction}

\section{Preliminaries}\label{sec:preliminaries}
\input{preliminaries}

\ignore{
\section{Proof overview}\label{sec:proof_overview}
\input{proof_overview}

}

\section{Parallel lower bounds for near-isotropic targets}\label{sec:near_isotropic_LB}
\input{near_isotropic_LB}

\section{Linear parallel-round lower bounds for anisotropic boxes}\label{sec:convex_anisotropic_LB}

\input{convex_anisotropic_LB}

\section{Conclusion and future directions}
\input{conclusion}

\input{appendix_simplified}

\paragraph{Acknowledgments.}
We thank Raghu Meka for helpful discussions throughout the project. We also thank Sitan Chen and Liye Wang for valuable discussions and insights.

\paragraph{AI methodology.} The authors have engaged in ongoing discussions with ChatGPT since
December 2025. The $\widetilde{\Omega}(d^{1/3})$ construction was
proposed by the authors, and ChatGPT assisted with formalizing and
completing the proof. The box construction yielding the $\Omega(d)$
lower bound arose during discussions of sampling from convex sets.
ChatGPT also assisted with exploratory discussions of possible reductions from optimization to sampling.
GPT-6 further assisted with revising the exposition.
The authors take full responsibility for the results and presentation
of this work.

\bibliographystyle{alpha}
\bibliography{refs}

\end{document}

%% file: macros.tex
\usepackage[T1]{fontenc}
\usepackage[margin=1in]{geometry}

\usepackage{latexsym}
\usepackage[small,bf]{caption2}
\usepackage{graphics}
\usepackage{epsfig}
\usepackage{amsopn}
\usepackage{url}
\usepackage{xcolor}
\usepackage{hyperref}
\hypersetup{hidelinks}

\usepackage{algorithm}
\usepackage{algpseudocode}
\usepackage{graphics}
\usepackage{comment}

\usepackage{tabularx}
\usepackage{booktabs}

\usepackage[shortlabels]{enumitem}

\usepackage{multirow}

\usepackage{titling}

\usepackage{amsmath}
\usepackage{amssymb}
\usepackage{amsthm}
\usepackage{hyperref}   
\usepackage[nameinlink,noabbrev]{cleveref}
\usepackage{aliascnt}

\usepackage{mathtools}
\usepackage{enumitem}
\usepackage{bbm}
\usepackage{turnstile}

\newtheorem{theorem}{Theorem}[section]

\newaliascnt{lemma}{theorem}
\newtheorem{lemma}[lemma]{Lemma}
\aliascntresetthe{lemma}

\newaliascnt{proposition}{theorem}

\aliascntresetthe{proposition}

\newaliascnt{fact}{theorem}

\aliascntresetthe{fact}

\newaliascnt{claim}{theorem}

\aliascntresetthe{claim}

\newaliascnt{definition}{theorem}

\aliascntresetthe{definition}

\newaliascnt{example}{theorem}

\aliascntresetthe{example}

\newaliascnt{corollary}{theorem}

\aliascntresetthe{corollary}

\newaliascnt{assumption}{theorem}

\aliascntresetthe{assumption}

\newaliascnt{observation}{theorem}

\aliascntresetthe{observation}

\newaliascnt{conjecture}{theorem}

\aliascntresetthe{conjecture}

\newaliascnt{question}{theorem}

\aliascntresetthe{question}

\newaliascnt{problem}{theorem}

\aliascntresetthe{problem}

\AddToHook{cmd/appendix/before}{\crefalias{section}{appendix}}

\crefname{theorem}{theorem}{theorems}
\Crefname{theorem}{Theorem}{Theorems}
\crefname{lemma}{lemma}{lemmas}
\Crefname{lemma}{Lemma}{Lemmas}
\crefname{proposition}{proposition}{propositions}
\Crefname{proposition}{Proposition}{Propositions}
\crefname{fact}{fact}{facts}
\Crefname{fact}{Fact}{Facts}
\crefname{claim}{claim}{claims}
\Crefname{claim}{Claim}{Claims}
\crefname{definition}{definition}{definition}
\Crefname{definition}{Definition}{definitions}
\crefname{example}{example}{examples}
\Crefname{example}{Example}{Examples}
\crefname{corollary}{corollary}{corollaries}
\Crefname{corollary}{Corollary}{Corollaries}
\crefname{assumption}{assumption}{assumptions}
\Crefname{assumption}{Assumption}{Assumptions}
\crefname{observation}{observation}{observations}
\Crefname{observation}{Observation}{Observations}
\crefname{conjecture}{conjecture}{conjectures}
\Crefname{conjecture}{Conjecture}{Conjectures}
\crefname{question}{question}{questions}
\Crefname{question}{Question}{Questions}
\crefname{problem}{problem}{problems}
\Crefname{problem}{Problem}{Problems}

\renewcommand{\epsilon}{\varepsilon}

\newcommand{\rr}{\mathbb{R}}

\newcommand{\poly}{\mathsf{poly}}

\renewcommand{\Pr}{\mathop{\bf Pr\/}}

\newcommand{\Unif}{\operatorname{Unif}}

\newcommand{\E}{\mathbb{E}}

\newcommand{\TV}{d_{\mathrm{TV}}}

\newcommand{\calN}{\mathcal{N}}

\newcommand{\bfS}{\mathbf{S}}

\newcommand{\ignore}[1]{{}}

%% file: introduction.tex
Diffusion models generate samples by progressively removing noise through repeated evaluations of a learned score function. In standard samplers such as DDPM \cite{HoEtAl20} and DDIM \cite{SongMengErmon21}, each step uses the output of the preceding step to determine the next query. To overcome this sequential bottleneck, recent parallel diffusion samplers based on Picard iteration, randomized midpoints, and speculative execution evaluate multiple denoising proposals simultaneously across future noise levels \cite{ShihEtAl23, TangEtAl24, ChenEtAl24, HuEtAl25, AnariEtAl26}. By trading parallel function evaluations for fewer sequential rounds of interaction, these methods raise a fundamental question:

\begin{center}
    \emph{How many sequential rounds of interaction with a learned score model are necessary, even when each round allows polynomially many evaluations?}
\end{center}

To address this question, we must account for how score functions are obtained. Because the target distribution $\pi$ is typically known only through finite training data, its exact
scores at different noise levels are unavailable. Instead, diffusion models learn score estimates
from noise-corrupted samples via a mean-square denoising score-matching objective \cite{SongErmon19, HoEtAl20}. Approximate scores satisfying this distributional $L^2$ guarantee alone support efficient sequential
sampling: for general targets with finite second moment, $\widetilde{O}(d)$ sequential score queries suffice to
sample suitably smoothed targets (or the target itself under mild regularity) \cite{ChenEtAl2023Sampling, BentonEtAl24,ChenEtAl2026HighAccuracy, LeeLuTan2023Convergence, LiYan2025Convergence}. 

In the parallel setting, however, existing guarantees rely on additional quantitative regularity assumptions \cite{ChenEtAl24, GCC25, ZS25} or stronger oracle guarantees, such as exact denoisers or uniform pointwise error control \cite{HuEtAl25, AnariEtAl26}. The closest diffusion lower bound, due to \cite{XP26}, shows that $\widetilde{\Omega}(\sqrt{d})$ total score queries are necessary under $L^2$ accuracy.  
Their construction uses Gaussian mixtures with a hidden number $n$ of centers. For any fixed query, the approximate oracle returns an uninformative response with high probability unless the codebook rate $\kappa=(\log n)/d$ lies in an interval of width $\widetilde O(1/\sqrt{d})$ determined by the query’s noise level. Hiding the rate among sufficiently many separated possibilities therefore forces $\widetilde{\Omega}(\sqrt d)$ queries. This search does not itself enforce sequential dependence, since queries targeting different rate intervals can be issued in parallel. To establish a parallel-round lower bound, we need a construction in which informative queries in later rounds depend on answers from earlier rounds.

\ignore{This dependence between successive evaluations creates a sampling bottleneck: each evaluation can exploit parallel hardware, but the next input typically depends on the preceding output. Recent methods address this bottleneck by evaluating proposed denoising steps in parallel \cite{ShihEtAl23, TangEtAl24, ChenEtAl24, HuEtAl25}. 
}

\paragraph{Oracle model.} We formalize this discussion through an oracle model that abstracts away training and network architecture as studied in \cite{ChenEtAl2023Sampling, BentonEtAl24, XP26}. For a target distribution \(\pi\) on \(\mathbb R^d\), let \(\pi^{(\sigma)}=\pi*\mathcal N(0,\sigma^2 I_d)\) denote its Gaussian smoothing and let \(s_{\pi,\sigma}=\nabla\log\pi^{(\sigma)}\) be the corresponding score function. An algorithm may query any location \(x\in\mathbb R^d\) and any noise level \(\sigma>0\). The oracle returns \(\widehat s_\sigma(x)\), where the family \(\widehat s\) is fixed before the interaction and satisfies
\begin{align}
    \mathbb E_{X\sim\pi^{(\sigma)}}\!\left[\left\|\widehat s_\sigma(X)-s_{\pi,\sigma}(X)\right\|_2^2\right]
\leq \frac{\varepsilon_{\mathrm{score}}^2}{\sigma^2}
\qquad\text{for every }\sigma>0.\label{eqn:L2_approx}
\end{align}
This mean-square guarantee reflects the distributional objectives used to learn scores. It controls error on noisy target samples while permitting larger errors in regions of small probability. In each adaptive round, the algorithm chooses a batch of queries before seeing any answers from that batch. Subsequent batches may depend on all previous answers. We allow polynomially many queries per round and unrestricted computation between rounds. The algorithm receives only the oracle answers and stated public information, and need not follow a prescribed diffusion trajectory.

\subsection{Main results}

In this work, we establish the first polynomial parallel-round lower bounds for diffusion sampling with approximate scores. Both results below concern sampling the target itself to total variation error at most \(1/10\). They hold for arbitrary randomized algorithms making polynomially many queries per round, with arbitrary query locations and noise levels, and with inverse-polynomial score error in \eqref{eqn:L2_approx}.

\begin{theorem}[Near-isotropic targets; consequence of \Cref{thm:parallel-lower-bound}]
For all sufficiently large \(d\), there is a family of smooth, near-isotropic Gaussian mixtures on \(\mathbb R^d\) for which sampling with $L^2$-approximate scores requires \(\widetilde\Omega(d^{1/3})\) adaptive rounds.
\end{theorem}

The lower bound rules out a general polylogarithmic-round sampling guarantee based only on near-isotropy and the mean-square score accuracy in \eqref{eqn:L2_approx}.
It is worth noting that the same target family remains sampleable under every legal oracle using $\widetilde{O}(d^{5/6})$ sequential
score queries, by a specialization of \cite{ChenEtAl2026HighAccuracy} proved in \Cref{lem:section33-upper-bound}. A polynomial gap remains
between these lower and upper bounds on the number of rounds.

While our first result shows that poor covariance conditioning is not necessary for parallel sampling hardness, anisotropy can provide an additional source of difficulty. Here anisotropy refers to widely different variances in different directions, measured by the condition number \(\kappa\) of the covariance matrix. Previous works such as \cite{LV06HitAndRun, JiaEtAl2026Rounding, ChewiEtAl24} have demonstrated this hardness and ways to address it in various sampling settings. 

\begin{theorem}[Anisotropic boxes; consequence of \Cref{thm:box-parallel}]
For all sufficiently large \(d\), there is a family of full-dimensional, axis-aligned boxes \(K\subseteq B_2(0,1)\) for which sampling from \(\operatorname{Unif}(K)\) with $L^2$-approximate scores requires \(\Omega(d)\) adaptive rounds.
\end{theorem}

The box family has substantial anisotropy: its coordinate scales decrease geometrically. In each coordinate, the box uses one interval selected from a public collection of disjoint candidates. The covariance matrix has condition number $\exp (\widetilde{\Theta}(d))$. Only the selected labels are hidden, and no interior point is supplied. 

For comparison, a direct recovery algorithm uses \(O(d)\) rounds
, matching the linear lower bound. This matching upper bound relies on the known candidate-interval geometry which we describe in more detail in \Cref{lem:box-upper}.

\paragraph{Proof overview.}
Our constructions arrange hidden information in an ordered sequence of blocks, following the successive-exposure approach used in parallel matroid search and sampling with counting oracles~\cite{KUW85,AGR24}. The technical challenge is to realize this successive exposure through a fixed oracle that remains \(L^2\)-accurate at every noise level while keeping the target blocks independent. For a product of independent hard instances, the exact score separates across blocks, so a single query probes all of them. We preserve the independent product structure but construct a fixed approximate oracle whose response in a later block depends on consistency tests involving all earlier blocks. These tests accept genuine noisy target samples with high probability. We bound the score error introduced when these tests fail and choose their failure probabilities small enough to ensure the mean-square accuracy guarantee at every noise level.

The main difficulty is to maintain this dependence when queries may choose arbitrary noise levels. In the near-isotropic construction, each block contains an independent random subset of cube vertices used as Gaussian centers, or codebook, and we use the one-block estimates of~\cite{XP26}. The codebook cardinalities increase across blocks, while their coordinate dimensions are equal. Sparsity and density compare the number of centers with the information retained after smoothing. At any queried noise level, either the first unknown codebook is sufficiently sparse that a query chosen without it rarely passes its gate, or all later codebooks are sufficiently dense that their local answers equal the public score of the smoothed uniform cube distribution with high probability. In the first case, a closed earlier gate suppresses every later contribution; in the second, those contributions are already absent. Thus, with high probability, later responses are exactly public. This exact equality allows us to compare transcripts even when oracle answers have arbitrary precision. This dichotomy applies separately to each query, including queries at different noise levels in the same batch.

Under the independent codebook law, a progressive simulation couples the actual execution, with small disagreement probability, to one whose responses in round $r$ depend only on the first $r$ codebooks and the algorithm's randomness. After $R$ rounds, the simulated output is independent of every unexposed codebook. A genuine sample passes a test associated with such a codebook with high probability, whereas an independent output rarely passes. The coupling transfers this distinction to the actual output and yields a constant total variation gap. Conditioning on near-isotropy preserves this bound.

Using more blocks strengthens the round lower bound, but leaves fewer coordinates per block and therefore requires larger gaps between adjacent codebook rates. Fitting these rates into a fixed interval yields the $1/3$ exponent. \Cref{subsec:exp_1_3} gives the parameter calculation.
For boxes, rapidly decreasing coordinate scales implement the same information requirement in individual coordinates, producing $d$ successive barriers at the cost of strong anisotropy.

\subsection{Related work}

\paragraph{Sequential diffusion sampling.} Polynomial-complexity guarantees under mean-square score error were established by \cite{ChenEtAl2023Sampling,LeeLuTan2023Convergence}, using early stopping to accommodate nonsmooth or singular targets. Subsequent work improved the dimension and accuracy dependence: \cite{BentonEtAl24} obtained a \(\widetilde O(d/\delta^2)\) bound, \cite{LiYan2025Convergence} obtained \(\widetilde O(d/\delta)\), and \cite{ChenEtAl2026HighAccuracy} obtained \(\widetilde O(d\,\operatorname{polylog}(1/\delta))\), with intrinsic-dimension refinements. Here \(\delta\) denotes sampling accuracy, and the bounds require sufficiently accurate scores. For general targets, strong-metric guarantees apply to Gaussian-smoothed distributions. Guarantees for the original target require a weak metric or additional regularity.

\paragraph{Parallel sampling algorithms.} ParaDiGMS \cite{ShihEtAl23} uses Picard iteration to refine multiple denoising steps simultaneously. Theoretical diffusion guarantees \cite{ChenEtAl24,GCC25,ZS25} accommodate mean-square score error together with quantitative regularity assumptions. For example, the round complexity in \cite{GCC25} depends on Lipschitz bounds for both the true and estimated scores. Autospeculation provides another route to parallel sampling \cite{HuEtAl25, AnariEtAl26}. The analyzed runtime guarantees use exact denoisers, with \cite{AnariEtAl26} also allowing uniformly bounded pointwise denoiser error. Parallel Langevin methods also give favorable dimension dependence for smooth targets under strong log-concavity or log-Sobolev assumptions \cite{SL19,ACV24}. These upper bounds depend on additional parameters beyond the mean-square score error in \eqref{eqn:L2_approx}.

\paragraph{Query and round lower bounds.} Lower bounds for exact first-order sampling include sharp one-dimensional results \cite{CGLGR22} and condition-number-dependent bounds already for Gaussian targets \cite{ChewiEtAl24}. Other work treats stochastic gradients \cite{CBL22}, relative Fisher information accuracy \cite{CGLL23}, or non-log-concave targets with value-and-gradient access \cite{HZ25}. For convex bodies, \cite{GRV15} gives a linear sequential membership-query lower bound using boxes, and \cite{Vem26} proves an \(\Omega(d^2/(\log d)^9)\) expected total-query lower bound for bodies satisfying \(B_2(0,1)\subseteq K\subseteq B_2(0,O(d^{3/2}))\). These bounds count queries rather than adaptive rounds. In the counting-oracle model, \cite{AGR24} proves a \(\widetilde\Omega(d^{1/3})\) round lower bound using sequential exposure, following the parallel-search framework of \cite{KUW85,KUW88}. Its oracle differs from smoothed-score access. In diffusion, \cite{XP26} supplies the total-query lower bound and codebook ingredients discussed above. Concurrent work \cite{CW26} compares corruption processes, obtaining \(\widetilde O(\sqrt d)\)-query sampling algorithms for random empirical measures on the Boolean hypercube with sufficiently accurate uniform or Gaussian diffusion scores, and a \(\widetilde\Omega(d)\) query lower bound for approximate masked diffusion. Our focus is the sequential dependence that persists when many score queries can be issued simultaneously.

\paragraph{Organization. } \Cref{sec:preliminaries} introduces our notation and oracle model. \Cref{sec:near_isotropic_LB} proves the parallel lower bound for near-isotropic Gaussian mixtures. \Cref{sec:convex_anisotropic_LB} establishes the linear lower bound for anisotropic boxes. Parameter choices and detailed proofs appear in \Cref{appendix_for_near_isotropic} and \Cref{app:boxes}.

%% file: preliminaries.tex
We formalize the score oracle and the parallel sampling model, then record
the Gaussian identities used in both constructions.

\paragraph{Notation.}
Unless specified otherwise, logarithms are natural. We write $[n]=\{1,\ldots,n\}$, let $I_d$ denote the identity matrix, and use $\|\cdot\|_2$, $\|\cdot\|_{\mathrm{op}}$ and $\|\cdot\|_{\mathrm{F}}$ for the Euclidean, operator and Frobenius norms. The Euclidean ball of radius $r$ centered at the origin is $B_2(0,r)$. We write 
$\operatorname{Unif}(S)$ for the uniform law on a nonempty finite set or a measurable set of positive finite volume, and identify distributions with their densities when these exist.

For a distribution $\pi$ with finite second moments, write
$\operatorname{Cov}(\pi)$ for its covariance matrix. In this paper, a distribution is isotropic if its mean is zero and its covariance is \(cI_d\) for some \(c>0\). We quantify near-isotropy through explicit bounds on the mean and on the operator-norm deviation of the covariance from a fixed positive constant multiple of \(I_d\). For a positive definite matrix $\Sigma$, its condition
number is $\kappa(\Sigma)=\lambda_{\max}(\Sigma)/\lambda_{\min}(\Sigma)$.

\subsection{Gaussian smoothing and approximate score access}
For a target $\pi$ on $\mathbb R^d$ and an additional noise level
$\sigma>0$, define
\[
  \pi^{(\sigma)}=\pi*N(0,\sigma^2I_d),
  \qquad
  s_{\pi,\sigma}(x)=\nabla\log\pi^{(\sigma)}(x).
\]
Thus $\pi^{(\sigma)}$ is the law of $Y+\sigma Z$ for independent
$Y\sim\pi$ and $Z\sim \calN(0,I_d)$. Gaussian convolution has a positive
smooth density, so the score $s_{\pi,\sigma}$ is defined at every query point. 

Fix $\varepsilon_{\mathrm{score}}>0$. A legal score oracle for $\pi$ is a fixed map $\widehat s:\mathbb R^d\times(0,\infty)\to\mathbb R^d$. Writing $\widehat s_\sigma(x)=\widehat s(x,\sigma)$, we require
\[
\mathbb E_{X\sim\pi^{(\sigma)}}
  \bigl\|\widehat s_\sigma(X)-s_{\pi,\sigma}(X)\bigr\|_2^2
  \le\frac{\varepsilon_{\mathrm{score}}^2}{\sigma^2}
  \qquad(\sigma>0).
\]
The error is averaged under the true noisy target. The oracle is fixed once the instance is chosen and remains unchanged throughout the interaction. We impose no pointwise accuracy guarantee or spatial Lipschitz bound, and \(\widehat s\) need not be the score of another distribution. An algorithm must succeed for every legal oracle for its target. A lower-bound construction may therefore assign one particular legal oracle to each target.

\subsection{Parallel algorithms and sampling accuracy}
An $R$-round algorithm submits at most $P$ queries $(x,\sigma)$ per round, choosing the entire batch before receiving any answers from that batch. Queries within a batch may use different locations $x$ and noise levels $\sigma$. Later batches may depend on all previous answers and on a random seed $\omega$ independent of the instance, containing all algorithmic and output randomness. The algorithm receives only the stated public information and oracle answers, and computation between rounds is unrestricted.

For a fixed target-oracle pair $(\pi,\widehat s)$, let $Q_{\pi,\widehat s}$ denote the algorithm's output law, averaging over $\omega$. We measure sampling error in total variation:
\[
  d_{\mathrm{TV}}(\mu,\nu)
  =\sup_A|\mu(A)-\nu(A)|,
\]
where the supremum is over measurable sets. Thus any measurable event gives a lower bound on total variation. We also use the coupling inequality
\[
    d_{\mathrm{TV}}(\mu,\nu)
    \le
    \Pr[X\ne Y],
\]
for any coupling $(X,Y)$ with marginal laws $\mu$ and $\nu$.

The sampling objective is to make $d_{\mathrm{TV}}(Q_{\pi,\widehat s},\pi)$ small. In our distributional lower bounds, the law of the hard instance is chosen independently of the algorithm. Probability over instances is separate from the algorithmic randomness already averaged in $Q_{\pi,\widehat s}$.

Throughout, we fix a polynomial query budget $P=P(d)$ and an inverse-polynomial score accuracy $\varepsilon_{\mathrm{score}}=\varepsilon_{\mathrm{score}}(d)$ before constructing the hard-instance family. The family and implicit constants may depend on these fixed choices.

\subsection{Gaussian identities}
Gaussian smoothing preserves product structure. Consequently, the score of a smoothed product distribution is the concatenation of the scores of its smoothed factors. 

We also use Tweedie's formula \cite{Robbins56, Efron11}. For bounded $Y\sim\nu$ and an independent $Z\sim\mathcal{N}(0,I_d)$,
\[
    x+\sigma^2 s_{\nu,\sigma}(x)
    =
    \mathbb{E}[Y\mid Y+\sigma Z=x].
\]
Here the posterior is defined by weighting $\nu$ by the Gaussian likelihood at $x$, so the identity holds at every point.

%% file: proof_overview.tex
Before we delve into these lower bound constructions, we briefly discuss one main ingredient underlying all these constructions. In the \(L^2\)-approximate oracle model, the central object of a lower-bound construction is not just a target distribution \(\pi\), but a pair consisting of \(\pi\) and a carefully chosen “fake score” \(\widehat{s}\). A sampling algorithm must succeed using only the prescribed approximation guarantees, so it suffices to construct one legal oracle that conceals the information needed to sample. Writing \(\pi^{(\sigma)}=\pi*\mathcal N(0,\sigma^2I_d)\), recall this oracle must satisfy the following at every noise level

\[
\mathbb E_{X\sim\pi^{(\sigma)}}\! \left[\left\|\widehat{s}_\sigma(X) -\nabla\log\pi^{(\sigma)}(X)\right\|_2^2\right] \le \frac{\varepsilon_{\mathrm{score}}^2}{\sigma^2}.
\]

Importantly, \(\widehat{s}\) is a single fixed measurable family, chosen once the target instance is fixed—not an adversary that changes its answers during the interaction. Our task is therefore to construct \emph{a fake oracle that is simultaneously analytically accurate and informationally unhelpful}, rather than to show that the exact score itself is uninformative.

In what follows, the diffusion lower bounds use the same principle: \emph{a fixed approximate oracle can remain accurate on noisy target samples while forcing an algorithm to discover hidden information in order}. Three claims make this principle useful:
\begin{enumerate}[(1)]
    \item The oracle must be legal for each instance.
    \item A parallel batch must rarely reveal information beyond the next hidden piece and
    \item An undiscovered piece must obstruct sampling
\end{enumerate}
We first explain this common argument in \Cref{subsec:from_resp_to_seq}, then show how the two constructions realize it and why their exponents differ in \Cref{sec:near-isotropic} and \Cref{subsec:overview_anisotropic_boxes}.

\subsection{From accurate responses to sequential information}\label{subsec:from_resp_to_seq}

\paragraph{Creating dependence between blocks.} The main idea is to hide independent pieces of information in coordinate blocks. Since their exact scores separate, a product construction alone allows one query to probe every block. To address this issue, we create dependence through the approximate response.

At a fixed noise level, let \(s_j^{\mathrm{pub}}\) be a public baseline field for block \(j\). For each block and noise level, we choose a local response \(b_j\) equal to either the entire true block-score field \(s_j\) or the entire public baseline field \(s_j^{\mathrm{pub}}\). Let \(a_i(x_i)\in[0,1]\) tests whether the query is consistent with the hidden structure of block \(i\). For \(x=(x_1,\ldots,x_B)\), the "fake oracle" at block $j$ has the form

\[
\widehat s_j(x)
=s_j^{\mathrm{pub}}(x_j)
+\left(\prod_{i<j}a_i(x_i)\right)
\bigl(b_j(x_j)-s_j^{\mathrm{pub}}(x_j)\bigr).
\]

The noise level and hidden instance are suppressed from the notation. Only the correction to the public score is gated. A closed earlier gate makes every later response public. There is no gate from block \(j\) itself, since one round is allowed to reveal its first unknown block. All functions are fixed before the interaction.

\paragraph{Accuracy and hiding use different inputs.} A genuine noisy sample carries evidence of its hidden block and passes the corresponding test with high probability while an uninformed query can fail the same test. When a public local response is used, it is already accurate. Otherwise the error comes from suppressing a correction when an earlier gate fails. Bounds on both the failure probabilities and correction sizes yield \eqref{eq:prelim-score-accuracy} for every fixed instance and every noise level.

Now for any fixed query chosen independently of the remaining hidden pieces, an earlier unknown gate blocks access or the later local responses are exactly public, except with small probability over those pieces. Exact equality removes all hidden information from those responses.

\paragraph{Turning one-batch hiding into sampling hardness.} To handle adaptivity, we compare the real execution with a progressive simulation that, in round \(r\), replaces all responses beyond block \(r\) by public ones. The simulated batch depends only on the algorithm's seed and the first \(r-1\) hidden pieces, leaving the remaining pieces independent. Suppose that, for every fixed exposed prefix and seed, such a batch disagrees with the full oracle with probability at most \(\delta_{\mathrm{batch}}\) over the remaining pieces. Coupling at the first disagreement gives

\[
\Pr[\text{real and simulated outputs differ}]
\leq R\delta_{\mathrm{batch}}.
\]

This probability averages over the hard instance and the algorithm's seed. After \(R\) rounds, the simulated output depends only on the first \(R\) hidden pieces. For an unseen piece, we exhibit an event that is likely under the target but unlikely for an output independent of that piece. This gives an average total variation separation, which survives the coupling error and implies a hard instance. The simulation is a proof device while the actual oracle remains fixed. 

\subsection{Near-isotropic mixtures: ordering informative noise regimes}\label{sec:near-isotropic}
For this overview, assume \(B\) divides \(d\) and set \(d_{\mathrm{blk}}=d/B\). The formal construction handles the remaining coordinates. 

In block \(j\), choose a uniformly random \(n_j\)-point subset \(S_j\) of the Boolean cube, independently across blocks, and smooth its uniform distribution by Gaussian noise of fixed variance \(\gamma^2\). These sets are \emph{the codebooks}. A target sample independently selects one center from each codebook and adds Gaussian noise. The block partition and cardinalities \(n_j\) are public. Only the centers are hidden. Concentration of codebook moments gives near-isotropy after a public scalar normalization.

\paragraph{The two one-block facts.} At a fixed total noise level \(\tau\), the analysis of \cite{XP26} supplies two regimes. A sparse codebook occupies too little of the ambient cube for an uninformed query to pass a likelihood-based consistency test, although a noisy sample from that codebook passes with high probability. A dense codebook has enough centers that its smoothed score is close in mean square to the public score obtained from the full cube. In the dense regime, with high probability over the codebook, the local response can be chosen equal to this public field everywhere. The construction returns the true local score whenever the public field is not sufficiently accurate, preserving legality for every instance. 

We let the public rates \(\kappa_j=(\log n_j)/d_{\mathrm{blk}}\) increase with \(j\). Let \(\kappa_-(\tau)\leq\kappa_+(\tau)\) be the sparse and dense thresholds, whose gap is at most \(W\), and separate adjacent rates by more than \(W\). If the first unknown rate \(\kappa_r\) is at most \(\kappa_-(\tau)\), its gate is difficult to pass. Otherwise, every later rate satisfies

\[
\kappa_j>\kappa_r+W>\kappa_-(\tau)+W\geq\kappa_+(\tau)
\qquad(j>r),
\]
so the later responses are public with high probability. This is why a query rarely obtains information beyond the first unknown block. The estimate holds for each fixed query and noise level over the random codebooks. Taking a union bound over the finitely many queries in a batch covers arbitrary simultaneous noise choices. 

\paragraph{Why the exponent is \(1/3\).}  The transition width is bounded by \(W\asymp\sqrt{\log d/d_{\mathrm{blk}}}\) in the polynomial-query, inverse-polynomial-error regime. We separate adjacent codebook rates by more than \(W\). Fitting \(B\) rates into a constant interval therefore requires, up to integer rounding,
\[
B\sqrt{\frac{\log d}{d_{\mathrm{blk}}}}\lesssim1,
\qquad d=Bd_{\mathrm{blk}},
\qquad\text{hence } B^3\log d\lesssim d.
\]
This permits \(B=\Theta((d/\log d)^{1/3})\) successive barriers. The exponent comes from the dimension needed to separate their noise regimes. We choose the rates in a range that is sparse at the target level \(\tau=\gamma\). The gate at that level then supplies the distinguishing event: a genuine sample passes, whereas an output independent of the next codebook rarely does. We apply the independence-based argument before conditioning on the high-probability near-isotropy event. \Cref{sec:near_isotropic_LB} and \Cref{app:one-block} give the construction and quantitative bounds.

\subsection{Anisotropic boxes: one barrier per coordinate}\label{subsec:overview_anisotropic_boxes}
For boxes, a hierarchy of spatial scales replaces the high-dimensional codebooks. In each coordinate, select uniformly one of \(M\) disjoint public candidate intervals, independently across coordinates. Their lengths and location scales decrease geometrically with the coordinate index. The selected intervals form the unknown box. No interior point is supplied: such a point would identify every selected interval.

The posterior-mean identity in \Cref{sec:score} splits each score coordinate into the public term \(-x_j/\sigma^2\) and an instance-dependent correction. At noise levels large relative to a later coordinate's scale, that correction can be discarded within the error budget. When the noise is small enough for the correction to matter, it is also small enough to test whether an earlier coordinate lies near its correct hidden interval. A genuine noisy sample passes this test with high probability. An uninformed point can lie near at most one of the \(M\) candidates, so it passes with probability at most \(1/M\).

This is the same earlier-gate-or-public-response dichotomy as for the mixtures. A batch of \(P\) queries bypasses an unknown coordinate with probability at most \(P/M\). The progressive simulation then has disagreement probability at most \(RP/M\) after \(R\) rounds. If \(R<d\), its output remains independent of an unexposed interval and lies in that interval with probability at most \(1/M\), averaged over the hidden label. A target sample always lies in the box. Choosing \(M\) sufficiently large compared with \(dP\) yields the linear lower bound.

Each barrier now uses one coordinate, allowing \(d\) barriers. The cost is substantial anisotropy: in the stated polynomial regime, the covariance condition number is \(\exp(\Theta(d\log d))\). \Cref{sec:convex_anisotropic_LB} and \Cref{app:boxes} supply the scales, gates, and error estimates.

\subsection{Other oracle models: preserving rounds in an optimization reduction}
Theorem~\ref{thm:LB_classical_informal} follows from a different mechanism. We turn a hard convex objective into a distribution concentrated near its minimizers, so a constant-accuracy sample gives a sufficiently accurate optimizer with constant probability. The key requirement is to simulate a batch of sampling queries without serializing the optimization queries. Starting from \cite{BubeckJLLS2019}, the relevant optimization lower bound forces \(T\) rounds to achieve error a sufficiently small constant multiple of \(T^{-1/2}\) in dimension \(n=\widetilde\Theta(T^2)\).

For convex bodies, we encode low objective values as large cross-sections in auxiliary coordinates. Their volumes amplify the preference for near-minimizers. After one initial objective-value query at the origin, each membership query costs one objective-value query. A thin lift, together with a bound on the objective gap at the origin, achieves the required concentration in total dimension \(\widetilde\Theta(T^2)\) and with outer-to-inner radius ratio \(O(\sqrt T)\), giving the square-root bound. Requiring a constant radius ratio costs more auxiliary dimensions: the total becomes \(\widetilde\Theta(T^{5/2})\), giving the \(d^{2/5}\) bound. These exponents record the dimension cost of converting optimization accuracy into sampling concentration.

For smooth, strongly log-concave targets, a globally Lipschitz convex extension penalizes leaving the optimization ball. Exponential reweighting favors low objective values. A quadratic term supplies \(1\)-strong convexity, and Gaussian smoothing supplies polynomially bounded smoothness. Projecting a sample back to the ball then gives an approximate minimizer. Approximate potential values and gradients are simulated using polynomially many Gaussian shifts chosen once and reused throughout the interaction. For each fixed objective, an integrated mean-square estimate makes the resulting fixed oracle legal with sufficiently high constant probability over this one-time choice. All shifted queries within a sampling batch are evaluated in parallel. This construction preserves dimension, giving the square-root lower bound. In both reductions, an \(R\)-round sampler yields an optimizer using \(R+O(1)\) rounds. The lifting, smoothing, and concentration estimates are proved in \Cref{sec:LB_other_query_models}, \Cref{app:membership} and \Cref{app:logconcave}.

%% file: near_isotropic_LB.tex
In this section, we prove that a sampler using too few rounds produces a distribution far from its target in total variation. Our construction uses $B$ coordinate blocks, each with a finite set
of centers $S_j$, called a \emph{codebook}. We write
$\bfS=(S_1,\ldots,S_B)$ for the full collection of codebooks
and reserve $S$ for a single codebook.

\begin{theorem}[Parallel lower bound for near-isotropic targets]\label{thm:parallel-lower-bound}
For all sufficiently large $d$, there is a family of near-isotropic Gaussian-mixture targets $\pi^{\bfS}$, each equipped with an approximate oracle \(\widehat s\) satisfying \eqref{eqn:L2_approx} such that every randomized sampler using $O\big((d/\log d)^{1/3}\big)$ rounds has an output law \(Q_{\pi^{\bfS},\widehat s}\) satisfying
\[
\Pr_{\bfS}[d_{\mathrm{TV}}(Q_{\pi^{\bfS},\widehat s},\,\pi^{\bfS})\geq 0.9]\geq0.9.
\]
\end{theorem}

To prove the theorem, we first construct the target distribution and verify the accuracy of its approximate oracle. \Cref{subsec:one_block} shows that each round reveals information about at most one additional block. \Cref{subsec:unseen_block_gives_TV} uses an unseen block to obtain a total-variation gap, and \Cref{subsec:exp_1_3} chooses the number of blocks to obtain the exponent \(1/3\).  We give the construction and main proof steps here, deferring parameter choices and auxiliary estimates to \Cref{appendix_for_near_isotropic}.

\paragraph{Parallel hardness.} Following \cite{KUW85, AGR24}, we seek a block construction in which each round reveals information about at most one additional block. We divide the $d$ coordinates into $B=\poly(d)$ blocks of dimension $m=\lfloor d/B\rfloor$ and plant independent hard instances within them. We arrange the instance sizes on an increasing scale to handle queries at different noise levels. The main obstruction is that \emph{the score function factorizes}: a global score query returns the score of every block at once. A parallel batch spanning different noise levels can therefore obtain information about many blocks simultaneously.

We address this obstruction by exploiting the approximate $L^2$-guarantee \eqref{eqn:L2_approx}: it suffices to construct a fixed approximate oracle $\hat{s}$ satisfying this guarantee and prove hardness for that oracle. We keep the target blocks independent but make the oracle's answer in each block depend on tests of the points supplied in earlier blocks. Each test checks whether the supplied point looks like a noisy sample from that block's hidden centers. If an earlier test returns zero, the later block receives a null score independent of its hidden centers. Genuine noisy samples usually pass these tests, allowing the oracle to maintain its $L^2$ accuracy. Together with the increasing instance sizes, this gate structure lets us prove that each round reveals information about at most one additional block. We now make this idea precise.

\begin{figure}[t]
  \centering
  \includegraphics[width=0.7\textwidth]{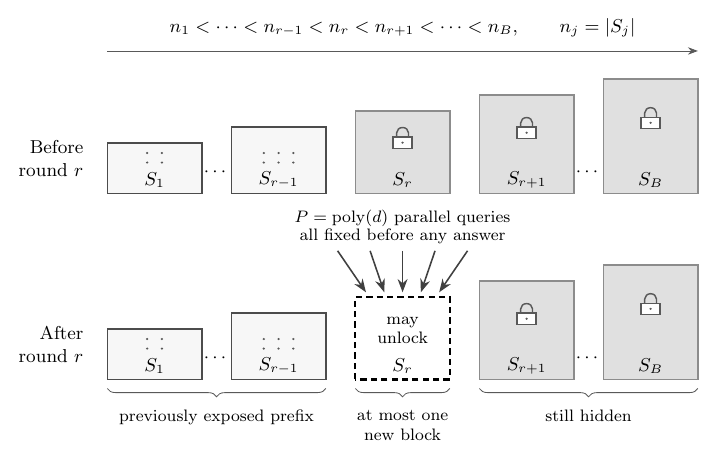}
  \caption{One-block-at-a-time exposure in the prefix-gated construction.
  The increasing heights schematically represent the codebook sizes
  $n_1<\cdots<n_B$. Here all coordinate blocks have the same dimension.
  Before round $r$, the progressive transcript depends only on the exposed
  prefix $S_1,\ldots,S_{r-1}$ and the algorithm's randomness.
  Even with polynomially many parallel queries, all locations and noise
  levels must be chosen before receiving any answer.
  With high probability, this batch can be answered using only $S_r$
  in addition to the exposed prefix, leaving $S_{r+1},\ldots,S_B$ hidden.
  The dashed block indicates possible new dependence, not guaranteed
  recovery of $S_r$.}
  \label{fig:parallel-hardness}
\end{figure}

\paragraph{Construction. } Let \(B\) be the number of blocks and \(m=\lfloor d/B\rfloor\) their dimension\footnote{Any remaining coordinates use independent uniform signs plus Gaussian noise.}. In block \(j\), choose a random set \(S_j\) of cube vertices in \(\{-1,+1\}^m\). Fix a constant target smoothing parameter $\gamma>0$. At total noise level \(\tau\ge\gamma\), sample each block independently as:

\begin{align}\label{ni:main_rates}
    X_j=Y_j+\tau Z_j,
\qquad
Y_j\sim\operatorname{Unif}(S_j),\quad
Z_j\sim\mathcal N(0,I_m).
\end{align}
Write \(\pi_\tau^{\mathbf S}\) for the resulting product law, including the padding coordinates, and \(\pi^{\mathbf S}=\pi_\gamma^{\mathbf S}\) for the target. Additional query noise $\sigma>0$ corresponds to total noise
$\tau=\sqrt{\gamma^2+\sigma^2}$, so
\[
    (\pi^{\mathbf S})^{(\sigma)}
    =
    \pi_{\sqrt{\gamma^2+\sigma^2}}^{\mathbf S}
    =
    \pi_\tau^{\mathbf S}.
\]
The target itself corresponds to $\tau=\gamma$, whereas legal oracle
queries have $\tau>\gamma$. The codebooks are independent across blocks. Each codebook is uniform among sets of its prescribed size, with these sizes increasing with block index \(j\). 

\ignore{
\subsection{A target with blockwise well-separated hidden structure}

The starting point is the way \cite{KUW85, AGR24} use blocks to limit the information revealed in each round. The common aim is to make one round of queries reveal information about only one additional block.

\cite{KUW85} consider a selection problem with hidden constraints. The elements are randomly divided into equal-sized blocks, each with a quota limiting how many elements can be selected. A query proposes a set and receives one bit indicating whether all quotas are respected. 

In the random construction, a query chosen without knowing the remaining blocks typically meets them in approximately equal numbers. \cite{KUW85} space the quotas far enough apart that blocks beyond the first unseen one rarely affect the answer. With high probability, all answers in a polynomial-sized batch can be determined using the previously revealed blocks and just one additional block. Therefore no information from later blocks is leaked in the round. After rounds fewer than the block number, an output by the algorithm is unlikely to fill all remaining quotas.

\cite{AGR24} use a related construction for exact counting. Again using a hidden random partition, they impose random linear equations on the binary variables in each block, with fewer equations in later blocks. A query fixes some variables and asks how many solutions remain. If too many unseen variables are fixed, the first unseen block typically has no completion, making the whole count zero. If fewer are fixed, later blocks typically contribute predictable counts whose product is independent of their hidden equations. Their threshold spacing again makes one additional block sufficient to simulate a round with high probability. 

Our construction follows this pattern using Gaussian mixtures. We separate the coordinates into blocks and hide a set of centers in each one. The numbers of centers increase across blocks, paralleling the increasing quotas in \cite{KUW85} and numbers of solutions in \cite{AGR24}. For an exact product score, the answer in each block depends only on the point supplied there. We construct an \emph{approximate oracle} whose answer in block $j+1$ also depends on a test of the point supplied in block $j$. The test checks whether this point looks like a noisy sample from block $j$'s hidden centers. When the test returns $0$, block $j+1$ receives a null score independent of its hidden centers. Thus, to obtain information about block $j+1$, a query must supply a point in block $j$ that opens the test. 

The $L_2$ approximate-score model \eqref{eqn:L2_approx} makes this test possible. Its accuracy guarantee averages over noisy target samples, which usually pass the test. We choose the codebook sizes so that, with high probability, a query chosen independently of the unseen codebooks either fails a test in the first unseen block or receives null answers in later blocks. Sections 4.2--4.3 explain why this holds at every noise level. The sequential obstruction comes from tests of the hidden centers, so the partition itself can remain fixed.

Let \(B\) be the number of blocks and \(m=\lfloor d/B\rfloor\) their dimension\footnote{Any remaining coordinates use independent uniform signs plus Gaussian noise.}. In block \(j\), choose a random set \(S_j\) of cube vertices in \(\{-1,+1\}^m\). Such a set is called a \emph{codebook}. A target sample at total noise level $\tau\geq\gamma$ chooses one of its vertices uniformly and adds Gaussian noise:
\begin{align}
    X_j=Y_j+\tau Z_j,
\qquad
Y_j\sim\operatorname{Unif}(S_j),\quad
Z_j\sim\mathcal N(0,I_m).\label{def:target}
\end{align}
The codebooks are independent across blocks. Each codebook is uniform among sets of its prescribed size increasing with block index \(j\). Write \(S=(S_1,\ldots,S_B)\) for the instance and \(\pi^S\) for its target distribution.
}

For one noise level $\tau$, let \(s_j\) be the true score in block \(j\). Let \(s^0\) be the null score of the full cube smoothed at the same noise level $\tau$. \footnote{All blocks have dimension $m$
and use the same full-cube reference distribution. Thus, at a
fixed noise level $\tau$, the null score is identical across
blocks, and we omit the block index.} By null score, we mean the score function of block $j$ when the underlying $S_j$ is $\{-1, +1\}^m$. Notice this null score is independent of the hidden codebook. As in \cite{XP26}, we choose the blockwise local answer \(b_j\) as follows:
\begin{equation}\label{eq:def_bj}
     b_j(x_j)=
 \begin{cases}
   s^0(x_j),
   &\displaystyle
    \mathbb E\bigl\|s^0(X_j)-s_j(X_j)\bigr\|_2^2
       \le \frac{\eta^2}{4\tau^2},\\[1.2ex]
   s_j(x_j),&\text{otherwise},
 \end{cases}
\end{equation}
where $\eta=\varepsilon_{\mathrm{score}}/(8B)$. Here the expectation is over the noisy block sample in
\eqref{ni:main_rates}. In words, $b_j$ is equal to \(s^0\) when it is sufficiently accurate and \(s_j\) otherwise. This choice is fixed for each codebook and noise level.

Given this, we now specify the earlier-block tests, called \emph{gates}, and show that the resulting answers satisfy the $L^2$ score-error guarantee \eqref{eqn:L2_approx} which is averaged over noisy target samples. The gate $a_i(x_i)\in[0,1]$ tests whether the point supplied in block $i$
looks like a noisy sample from one of its hidden centers. The test compares
the likelihood of the point under each center with its likelihood under
the public full-cube distribution. The gate equals one when some center
passes the test, equals zero when every center clearly fails, and
interpolates between these values near the threshold.

We choose the threshold so that genuine noisy samples almost always pass.
Specifically, for every fixed codebook $S_i$, a sample $X_i$ generated as in
\eqref{ni:main_rates} satisfies
\[
  \Pr\bigl[a_i(X_i)<1\bigr]\le\zeta,
  \qquad
  \zeta\le\frac{\eta^2\tau^2}{16m}.
\]
Here $\zeta=\zeta(\tau)$ is a small public failure budget. \Cref{ni:accuracy-section}
gives the precise test and proves this guarantee. This acceptance property
is what we need for the following accuracy argument.

\paragraph{The gated oracle.} At the $j$-th block, the oracle returns

\begin{align}
    \widehat s_j(x)=s^0(x_j)
+\left(\prod_{i<j}a_i(x_i)\right)
\bigl(b_j(x_j)-s^0(x_j)\bigr). \label{def:prefix-oracle}
\end{align}
On the $d-Bm$ padding coordinates, we use their exact public score. Notice that passing all earlier tests gives the local answer $b_j$, while a closed earlier gate ($a_i(x_i) = 0$ for some $i < j$) leaves the null score $s^0$. The first fact we need is that their combination satisfies the required accuracy guarantee.

\begin{lemma}[Accuracy of the gated oracle]\label{eq:prefix-oracle}
For the target distribution constructed in \eqref{ni:main_rates} at $\tau = \gamma$, the oracle in \eqref{def:prefix-oracle} satisfies \eqref{eqn:L2_approx} at every additional query noise level $\sigma > 0$ with $\tau = \sqrt{\gamma^2 + \sigma^2}$.
\end{lemma}

\begin{proof}
Fix the codebooks and a total noise level $\tau$. By the acceptance
guarantee above, each genuine block sample $X_i$ satisfies
\[
  \Pr\bigl[a_i(X_i)<1\bigr]\le\zeta.
\]
By Tweedie's formula, we have
\[
  \sup_x\bigl\|s_j(x)-s^0(x)\bigr\|_2
  \le \frac{2\sqrt m}{\tau^2},
\]
because both posterior means lie in $[-1,1]^m$.

There are two cases in block $j$. If $b_j\equiv s^0$, then the gated
oracle also returns $s^0$, and its expected squared error is at most
$\eta^2/(4\tau^2)$ by the definition of $b_j$. Otherwise $b_j\equiv s_j$.
The error then vanishes whenever all preceding gates equal one, and its
squared norm is at most $4m/\tau^4$. A union bound over the preceding
gates therefore gives, in either case,
\[
  \mathbb E\bigl\|\widehat s_j(X)-s_j(X_j)\bigr\|_2^2
  \le \frac{\eta^2}{4\tau^2}
       +\frac{4m(j-1)\zeta}{\tau^4}.
\]
Summing over the blocks, with no error from the padding coordinates,
we obtain
\[
\begin{aligned}
  \mathbb E_{X\sim\pi_\tau^{\bfS}}
    \bigl\|\widehat s(X)-\nabla\log\pi_\tau^{\bfS}(X)\bigr\|_2^2
  &\le \frac{B\eta^2}{4\tau^2}
       +\frac{2mB(B-1)\zeta}{\tau^4}\\
  &\le \frac{B(B+1)\eta^2}{8\tau^2}\\
  &\le \frac{\varepsilon_{\mathrm{score}}^2}{256\tau^2}
   \le \frac{\varepsilon_{\mathrm{score}}^2}{\sigma^2}.
\end{aligned}
\]
Here we used $\zeta\le\eta^2\tau^2/(16m)$,
$\eta=\varepsilon_{\mathrm{score}}/(8B)$, and
$\tau^2=\gamma^2+\sigma^2\ge\sigma^2$.
\end{proof}

Moreover, one can show that the constructed oracle is globally Lipschitz in its spatial argument, with Lipschitz constant \(O_\gamma(d)\) uniformly over all query noise levels. See \Cref{ni:accuracy-section}.


\subsection{Only one additional block per round}\label{subsec:one_block}

Having established the oracle's accuracy, we proceed in two steps. First, we show that a batch based on previously exposed codebooks and an independent random seed can be answered using at most one additional codebook with high probability. We then extend this guarantee to adaptive algorithms through a simulation.

\paragraph{Sparse and dense blocks.} Suppose the first $r - 1$ codebooks have been exposed. In \eqref{def:prefix-oracle}, all blocks after $r$ return null scores if block $r$’s gate is closed or all later local answers are null. The former is likely when block $r$ is sufficiently sparse at the queried noise level while the latter is likely when all later blocks are sufficiently dense. We separate the codebook sizes so that \emph{every noise level falls into one of these two cases}.

To make this precise, we compare each codebook's size with the information retained after smoothing as in \cite{XP26}:

\begin{equation}
\underbrace{\kappa_j=\frac{\log|S_j|}{m}}_{\text{codebook rate}}
\qquad\text{and}\qquad
\underbrace{\alpha_\tau=\operatorname{MI}(Y^0;Y^0+\tau Z)}_{\text{information per coordinate}}.\label{eq:alpha-definition}
\end{equation}
Here $\operatorname{MI}(U;V)$ denotes the mutual information between $U$ and $V$. $Y^0 \sim \Unif (\{\pm 1\})$ is a uniform sign, \(Z\sim\mathcal N(0,1)\) is independent, and \(\tau\geq\gamma\) is the total Gaussian noise. The rates \(\kappa_j\) stay fixed as \(\tau\) changes, while \(\alpha_\tau\) decreases.

Crucially, this information level $\alpha_\tau$ determines the transition. Let \(\varphi_\tau\) be the density of \(\mathcal N(0,\tau^2 I_m)\) and \(q_\tau\) the density of the smoothed full cube. For a center \(y\) and observation \(x\), the log-likelihood ratio
\[
\ell_\tau(y,x)=\log\frac{\varphi_\tau(x-y)}{q_\tau(x)}
\]
measures the evidence for that center relative to the full cube. At a genuine noisy observation of \(y\), this quantity concentrates around its mean  $m \alpha_{\tau}$. The gate therefore tests whether some center supplies nearly this much evidence. 

Now fix a query \((x,\tau)\) chosen independently of the hidden codebook \(S_j\). To open the gate, this point must supply enough evidence for at least one of the codebook's centers. First consider the cutoff \(m\alpha_\tau\), before allowing for concentration margins. Averaging \(\exp(\ell_\tau(y,x))\) over a uniform cube vertex gives one, so Markov's inequality bounds the probability that a random center meets this cutoff by \(\exp (-m\alpha_\tau)\). Taking a union bound over the \(|S_j|=\exp (m\kappa_j)\) centers gives
\[
\begin{aligned}
\Pr_{S_j}\!\left[\exists y\in S_j:\ell_\tau(y,x)\geq m\alpha_\tau\right]
&\leq |S_j|\exp(-m\alpha_\tau)= \exp(-m(\alpha_\tau-\kappa_j)).
\end{aligned}
\]
When \(\kappa_j\) is sufficiently below \(\alpha_\tau\), this probability is small: there are too few centers to compensate for how unlikely each individual match is. Thus a sparse codebook has too few centers for an independent query to find a match and rarely opens the gate.

For a dense block,  sufficiently many random centers average out after smoothing: when \(\kappa_j\) is above \(\alpha_\tau\) by a suitable margin, the null score is accurate with high probability. The estimates in \Cref{ni:analytic} first control this approximation in
relative entropy at a slightly smaller noise level, then convert it into a bound on the mean-square
error of the null score. Markov’s inequality shows that this error falls below the threshold in \eqref{eq:def_bj}
with high probability. On that event, the local answer $b_j$ equals the null score at every query point.

The following lemma summarizes these guarantees for a small failure budget $\delta$.

\begin{lemma}[Sparse and dense blocks]\label{lem:sparse_and_dense}
    With the parameter choice of \Cref{ni:parameters}, there are thresholds $\kappa^{-}(\tau) \leq \alpha_{\tau} \leq \kappa^+ (\tau)$ with gap $O(\sqrt{\log d / m})$ uniformly over $\tau \geq \gamma$. For each fixed $\tau \geq \gamma$, a random codebook $S_j$ from our construction satisfies:
    \begin{enumerate}
        \item (Sparse blocks) If $\kappa_j \leq \kappa^- (\tau)$, then $\Pr_{S_j}[a_j(x_j) > 0] \leq \delta$ for every fixed point $x_j$.
        \item (Dense blocks) If $\kappa_j \geq \kappa^+ (\tau)$, then $\Pr_{S_j}[b_j \equiv s^0] \geq 1 - \delta$, where the equality holds as functions of the query point.
    \end{enumerate}
\end{lemma}

The precise thresholds and proofs are given in \Cref{ni:simulation-section} (\Cref{ni:hiding} and \Cref{ni:window}). We choose
adjacent rates to be separated by more than this width. Consequently, at every noise level, either
block $r$ is sparse or every later block is dense. \Cref{subsec:exp_1_3} explains how many such rates can be packed into the available range.

\paragraph{From one batch to adaptive rounds.}

We now extend the one-batch guarantee to adaptive algorithms. In round \(r\), the simulation returns the original gated answers from \eqref{def:prefix-oracle} in the first \(r\) blocks and null scores in all later blocks, retaining the exact public padding response. Consequently, after \(R\) rounds its output depends only on the first \(R\) codebooks and the algorithm's random seed. The following lemma says that this restricted simulation can reproduce the actual output of an adaptive algorithm with high probability.

\begin{lemma}[One new block per round]\label{lm:one-block-per-round}
Let a sampler use \(R<B\) rounds within the prescribed polynomial query budget. Its actual output \(Y\) and a simulated output \(\widetilde Y\) can be generated together so that
\[
\Pr[Y\ne\widetilde Y]\leq0.001.
\]
The simulated output depends only on the algorithm's random seed and \(S_1,\ldots,S_R\). The probability averages over the independent codebooks and seed.
\end{lemma}

\begin{proof}[Proof sketch] Run both interactions with the same codebooks and random seed. Fixing \(S_1,\ldots,S_{r-1}\) and the seed fixes the simulation’s round-\(r\) batch, including its noise levels, while the remaining codebooks retain their independent laws. By \Cref{lem:sparse_and_dense} and the rate separation, the full and simulated responses disagree on this batch with probability at most \(PB\delta\), where \(P\) bounds the number of queries per round. If the interactions first disagree in round \(r\), they have selected this same batch. Summing over rounds therefore gives \(\Pr[Y\ne\widetilde Y]\le RPB\delta\le0.001\), by the parameter choice in \Cref{ni:parameters}. The detailed argument appears in \Cref{ni:simulation-section} (\Cref{ni:one-batch} and \Cref{ni:progressive}).
    
\end{proof}



\subsection{An unseen block gives a TV gap}\label{subsec:unseen_block_gives_TV}
We now turn the restriction on information revealed per round into a lower bound on sampling error, using an unseen block to distinguish the sampler’s output from the target distribution. 

Fix a sampler using \(R\) rounds, where \(R<B\) and \(B\) is the number of blocks. By \Cref{lm:one-block-per-round}, its output can be coupled to a simulated
output independent of the codebook $S_{R+1}$. We distinguish the simulated output from a genuine
sample using the gate in block $R+1$ at the target noise level $\gamma$: a genuine sample passes with high probability, while a point independent of the hidden codebook rarely does. 

Let \(t_{\tau}\) be the log-likelihood cutoff of the gate at total noise \(\tau\), chosen slightly below \(m\alpha_{\tau}\) as in equation \eqref{ni:cutoff} in \Cref{ni:accuracy-section}. For a codebook \(S\subseteq\{\pm1\}^m\), define the target-level test:
\[
E(S)=\left\{x\in\mathbb R^m:
\max_{y\in S}\ell_\gamma(y,x)\geq t_\gamma\right\}.
\] 
Intuitively, \(E(S)\) consists of points close to a codebook center, where closeness is measured by the likelihood test above. The following lemma shows that genuine samples usually belong to \(E(S)\), whereas points chosen independently of the codebook $S$ rarely do.

\begin{lemma}[Recognizing samples from a hidden block]\label{lm:hidden}
Let \(n\) be one of the codebook sizes used in our construction. For every fixed \(n\)-element codebook \(S\), a sample from its target block \eqref{ni:main_rates} at $\tau = \gamma$ belongs to \(E(S)\) with probability at least \(0.999\). If \(S\) is uniform among the \(n\)-element subsets of \(\{\pm1\}^m\), then for every fixed \(x\),
\[
\Pr_S[x\in E(S)]\leq0.001.
\]
\end{lemma}

The first guarantee is the gate-acceptance property used in \Cref{eq:prefix-oracle}. For the second, our codebook sizes are chosen so that every block is sufficiently sparse at the target noise level \(\gamma\), allowing us to apply the sparse-block calculation from \Cref{subsec:one_block}. \Cref{ni:target-test} gives the quantitative bounds for the parameters chosen in \Cref{ni:parameters}.

Now we are ready to prove the total variation lower bound in our main theorem.
\begin{proof}[Proof of \Cref{thm:parallel-lower-bound}]
We first work under the unconditioned independent codebook law. Fix a sampler using $R<B$ rounds, and write $Q_{\bfS}$ for its output
law on instance $\bfS = (S_1,\ldots,S_B)$. For $x\in\mathbb{R}^d$, let
$x_{R+1}\in\mathbb{R}^m$ denote its coordinates in block $R+1$.
Define
\[
    D_{\bfS}
    =
    \{x\in\mathbb{R}^d:x_{R+1}\in E(S_{R+1})\}.
\]
This event checks whether the output passes the target-level test
in the first unseen block.

By \Cref{lm:hidden},
\[
    \pi^{\bfS}(D_{\bfS})\ge 0.999
\]
for every fixed instance $\bfS$. Couple the actual output $Y$ to the
simulated output $\widetilde Y$ as in \Cref{lm:one-block-per-round}. Conditional on
$S_1,\ldots,S_R$ and the algorithm's random seed,
$\widetilde Y_{R+1}$ is fixed, while $S_{R+1}$ remains an independent
random codebook. The second part of \Cref{lm:hidden} therefore gives,
after averaging,
\[
    \Pr[\widetilde Y\in D_{\bfS}]\le 0.001.
\]
Here the probability averages over the independent codebooks and
the algorithm's randomness. Consequently,
\[
\begin{aligned}
    \mathbb{E}_{\bfS}[Q_{\bfS}(D_{\bfS})]
    &= \Pr[Y\in D_{\bfS}] \\
    &\le
       \Pr[\widetilde Y\in D_{\bfS}]
       +\Pr[Y\neq\widetilde Y] \\
    &\le 0.002.
\end{aligned}
\]
For every fixed $\bfS$, total variation is at least the difference
between the probabilities of $D_{\bfS}$. Averaging gives
\begin{equation}
    \mathbb{E}_{\bfS}\!\left[
        d_{\mathrm{TV}}(Q_{\bfS},\pi^{\bfS})
    \right]
    \ge
    \underbrace{0.999}_{\text{genuine sample}}
    -
    \underbrace{0.001}_{\text{simulated output}}
    -
    \underbrace{0.001}_{\text{simulation error}}
    =
    0.997.
\end{equation}

Let \(G\) be the event in \Cref{ni:near-isotropy}, on which the targets are near-isotropic. For sufficiently large \(d\), \(\Pr(G)\ge1/2\). 
Since $1-d_{\mathrm{TV}}(Q_{\bfS},\pi^{\bfS})$ is nonnegative, the preceding
bound implies
\[
    \mathbb{E}_{\bfS}\!\left[
        1-d_{\mathrm{TV}}(Q_{\bfS},\pi^{\bfS})
        \,\middle|\,G
    \right]
    \le
    \frac{0.003}{\Pr(G)}
    \le 0.006.
\]
Markov's inequality now gives
\[
     \Pr_{\bfS}\!\left[
        d_{\mathrm{TV}}(Q_{\bfS},\pi^{\bfS})\ge 0.9
        \,\middle|\,G
    \right] = 1 - \Pr_{\bfS}\!\left[
        d_{\mathrm{TV}}(Q_{\bfS},\pi^{\bfS})<0.9
        \,\middle|\,G
    \right]
    \ge 1 - 
    \frac{0.006}{0.1} > 0.9.
\]

Choose the hard-instance distribution by conditioning the codebook law on \(G\). This distribution is independent of the sampler and is supported on near-isotropic targets equipped with legal oracles. The preceding probability bound, with \(B=\Theta((d/\log d)^{1/3})\), proves the theorem.
\end{proof}

\subsection{Why the exponent is $1/3$}\label{subsec:exp_1_3}

The preceding argument gives a lower bound of $B$ rounds, provided the codebook rates are sufficiently separated. We now explain how this requirement permits the choice $B = \Theta((d/ \log d)^{1/3})$. Adding
blocks strengthens the lower bound but leaves fewer coordinates in each block, increasing the
separation required between adjacent rates.

We choose the rates $\kappa_j$ in a fixed interval strictly inside
$(0,\alpha_\gamma)$. Keeping them a fixed distance below $\alpha_\gamma$
ensures that every block is sufficiently sparse at the target noise level,
as required in \Cref{subsec:unseen_block_gives_TV}. By \Cref{lem:sparse_and_dense}, the transition between the sparse
and dense regimes has width $O(\sqrt{\log d/m})$, uniformly over the noise
level. We therefore separate adjacent rates by a sufficiently large constant
multiple of $\sqrt{\log d/m}$. The logarithmic factor accounts for the
polynomial query budget and the prescribed inverse-polynomial score accuracy.

Such a choice is feasible when the required spacings fit inside the fixed
rate interval:
\[
  \underbrace{B}_{\text{number of blocks}}
  \cdot
  \underbrace{\sqrt{\frac{\log d}{m}}}_{\text{rate spacing}}
  \;\lesssim\;
  \underbrace{1}_{\text{available rate range}}.
\]
Equivalently, this choice of spacing requires block dimension
$m\gtrsim B^2\log d$. Since $m=\lfloor d/B\rfloor$ and $d\asymp Bm$,
we obtain
\[
  B^3\log d\lesssim d.
\]
Thus we may choose
\[
  B=\Theta\!\left((d/\log d)^{1/3}\right),\footnote{A similar cubic tradeoff underlies the $\widetilde{\Omega}(d^{1/3})$ lower bounds for parallel matroid search \cite{KUW85} and discrete sampling with counting oracles \cite{AGR24}: concentration estimates require block sizes quadratic in the number of blocks. Here the constraint $m\gtrsim B^{2}\log d$ comes from the separation needed between adjacent codebook rates.}
\]
with a sufficiently small constant factor. \Cref{ni:parameter-lemma} realizes this choice
with integer codebook sizes and verifies the remaining error bounds.
This supplies the parameter choice used in the proof of \Cref{thm:parallel-lower-bound}.

\paragraph{Upper bound for the same family.}
For comparison, the same target family admits a sequential sampler under
the same approximate-score guarantee, by specializing the algorithm
of~\cite{ChenEtAl2026HighAccuracy}.

\begin{lemma}[Upper bound for the same target family]
\label{lem:section33-upper-bound}
For all sufficiently large $d$, every target in the constructed family can
be sampled to total variation error at most $1/10$ using
$\widetilde O(d^{5/6})$ score queries, with one query per round, for every
fixed legal oracle.
\end{lemma}

The block product structure controls the regularity of the true denoiser,
the posterior-mean map: after a Gaussian observation, the posterior still
factorizes across blocks. This yields a query bound of
$\widetilde O(\sqrt{dm})=\widetilde O(d^{5/6})$ for our choice
$m=\widetilde\Theta(d^{2/3})$. \Cref{ni:upper} gives the proof and the complete
error and query bounds.

\ignore{
The same balance appears in \cite{KUW85} and \cite{AGR24}. The table compares their choices with ours, using \(B\) blocks of size \(m\), so that the total dimension or number of elements is of order \(Bm\).
\begin{center}
\begingroup
\small
\setlength{\tabcolsep}{4pt}
\renewcommand{\arraystretch}{1.22}
\begin{tabularx}{\linewidth}{@{}>{\raggedright\arraybackslash}p{0.13\linewidth}>{\raggedright\arraybackslash}p{0.24\linewidth}>{\raggedright\arraybackslash}X>{\raggedright\arraybackslash}p{0.19\linewidth}@{}}
\toprule
Construction & Parameter increased across blocks & Source of the square-root scale & Spacing used\\
\midrule
\cite{KUW85} & Selection limit & Random partition: \(\sqrt m\)-scale fluctuations in the number of selected elements in an unseen block. & \(\sqrt{m\log d}\) \\
\addlinespace
\cite{AGR24} & Number of variables minus number of equations & Random partition: \(\sqrt m\)-scale fluctuations in the number of fixed coordinates in an unseen block. & \(\sqrt{m\log d}\) \\
\addlinespace
Ours & Logarithm of the number of centers & Concentration of log-likelihood ratios under Gaussian smoothing gives a \(\sqrt m\) scale for separating sparse and dense codebooks. & \(\sqrt{m\log d}\) \\
\bottomrule
\end{tabularx}
\endgroup
\end{center}
}

%% file: convex_anisotropic_LB.tex
In this section, we prove an $\Omega(d)$ parallel-round lower bound for
uniform distributions over boxes and provide a matching upper bound. We give the construction and main proof steps here, deferring parameter choices and detailed estimates to \Cref{app:boxes}.

\begin{theorem}[Linear parallel lower bound]
\label{thm:box-parallel}
For sufficiently large $d$, there is a distribution over full-dimensional
axis-aligned boxes $K^\theta\subset B_2(0,1)$, with a fixed legal approximate oracle satisfying \eqref{eqn:L2_approx} for each
$\pi^\theta=\Unif(K^\theta)$, such that every randomized sampler using fewer than
$d$ rounds has output law $Q^\theta$ satisfying
\[
  \Pr_\theta\!\left[\TV(Q^\theta,\pi^\theta)\ge 0.9\right]\ge 0.9.
\]
\end{theorem}

Here, the instance distribution and oracle family are independent of the sampler.
The probability is over the instance. As in \Cref{sec:near_isotropic_LB}, the construction forces each round to expose information about at most one additional coordinate. 

Each coordinate hides one of many disjoint candidate intervals. The coordinate scales decrease rapidly, so that at large noise levels the oracle can suppress all later corrections within its accuracy budget. At smaller noise levels, a query must pass the proximity tests for earlier selected intervals to obtain a nonpublic answer in a later coordinate. The disjoint candidate neighborhoods make these tests difficult to pass without knowing the corresponding labels. 

In what follows, we first construct the box family and verify the accuracy of its oracle. We then show that each round exposes at most one additional coordinate and use an unseen coordinate to obtain the TV lower bound. Finally, we give a matching recovery algorithm.

\paragraph{Construction.}
Let $M$ denote the number of candidate intervals in each coordinate, chosen to be a sufficiently large polynomial in $d$.
In coordinate $i$, let $a_i>0$ denote their location scale. We place $M$ disjoint intervals
$I_i(1),\ldots,I_i(M)$ in $[a_i,2a_i]$, with equally spaced centers and lengths
much smaller than the gaps between them. The location scales and interval
lengths decrease rapidly with $i$. We defer the parameter choices to \Cref{app:box-parameters}.

Choose $\theta_1,\ldots,\theta_d$ independently and uniformly from $[M]$, and set
\begin{equation}
  K^\theta=\prod_{i=1}^d I_i(\theta_i),
  \qquad \pi^\theta=\Unif(K^\theta).
  \label{eq:box-target}
\end{equation}
All candidate intervals are available to the algorithms. Here only the labels are hidden, and no point
inside the box is supplied. The parameters are chosen so that every box lies
in the unit ball. The rapidly decreasing interval lengths give widely different coordinate variances, making the construction anisotropic.

\paragraph{The approximate oracle.}
Fix an instance $\theta$. For each coordinate $j$ and noise
level $\sigma>0$, let $m_{j,\sigma}^{\theta_j}(z)$ denote the
posterior mean of $U\sim\operatorname{Unif}(I_j(\theta_j))$
given $U+\sigma Z=z$, where $Z\sim N(0,1)$ is independent.
The superscript $\theta_j$ records the selected interval
in coordinate $j$.

Since Gaussian smoothing preserves the product structure,
Tweedie's formula gives
\begin{equation}
    s_{j,\sigma}^{\theta}(x)
    = -\frac{x_j}{\sigma^2}
      + \frac{m_{j,\sigma}^{\theta_j}(x_j)}{\sigma^2},
    \qquad
    m_{j,\sigma}^{\theta_j}(x_j)
    \in I_j(\theta_j)\subset[a_j,2a_j].
\end{equation}
We use the public null score $s_\sigma^0(z)=-z/\sigma^2$
and gate only the posterior-mean correction
$m_{j,\sigma}^{\theta_j}(x_j)/\sigma^2$.
Let $t_1>\cdots>t_d>0$ denote the noise cutoffs:
above $t_j$, the correction in coordinate $j$ is small
enough to discard within the score-error budget.

For each coordinate $i<d$ and candidate label $k\in[M]$,
we construct a gate
$g_{i,\sigma}^k:\mathbb R\to[0,1]$
that tests proximity to the interval $I_i(k)$:

\begin{itemize}
    \item When $\sigma\le t_{i+1}$, the gate equals one on a fixed
    neighborhood of $I_i(k)$, decreases linearly with distance
    from the interval, and equals zero outside a slightly
    larger neighborhood.
    We choose these neighborhoods so that
    $g_{i,\sigma}^k(U+\sigma Z)=1$ with high probability
    for $U\sim\operatorname{Unif}(I_i(k))$ and independent
    $Z\sim N(0,1)$.
    At the same time, the larger neighborhoods are disjoint:
    for any fixed input $z$, at most one label $k$ can satisfy
    $g_{i,\sigma}^k(z)>0$.
    \item When $\sigma>t_{i+1}$, set $g_{i,\sigma}^k\equiv1$.
    At these noise levels, every later coordinate $j>i$
    already satisfies $\sigma>t_j$, so its correction is
    discarded and the gate is unnecessary.
\end{itemize}

\Cref{app:box-parameters} specifies the neighborhoods and their Gaussian tail bounds. For the instance $\theta$, we use the gate
$g_{i,\sigma}^{\theta_i}$ corresponding to the selected
interval $I_i(\theta_i)$, and define
\begin{equation}
  \widehat s^\theta_{j,\sigma}(x)
  =s^0_\sigma(x_j)
   +\mathbf{1}_{\{\sigma\le t_j\}}
     \left(\prod_{i<j}g^{\theta_i}_{i,\sigma}(x_i)\right)
     \frac{m^{\theta_j}_{j,\sigma}(x_j)}{\sigma^2}.
  \label{eq:box-gated-oracle}
\end{equation}
Just as in the construction in \Cref{sec:near_isotropic_LB}, a closed earlier gate removes every later hidden correction.
The oracle is fixed once \(\theta\) is chosen. We need the following fact regarding the accuracy of the constructed oracle.

\begin{figure}[t]
\centering
\includegraphics[width=0.6\textwidth]{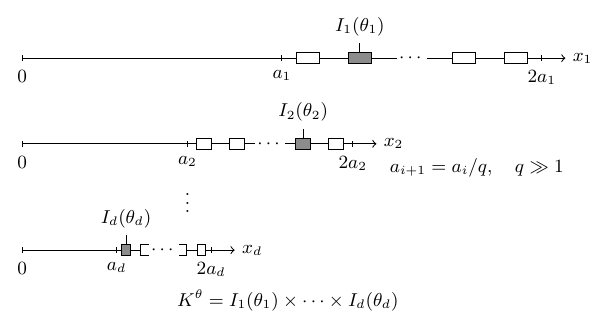}











\caption{
The anisotropic box construction. In coordinate $i$, the public
candidate intervals lie in $[a_i,2a_i]$, and the shaded interval
$I_i(\theta_i)$ is selected by an independent uniform label $\theta_i$.
The target is uniform on their Cartesian product. Both the coordinate
scales and the interval lengths decrease by a factor of $q$ at each
step. Only representative candidates are shown; scales are schematic.
}
\label{fig:box-construction}
\end{figure}

\begin{lemma}[Accuracy of the gated oracle]
\label{lem:box-accuracy}
For every instance $\theta$, the oracle in~\eqref{eq:box-gated-oracle} satisfies
the score-error guarantee \eqref{eqn:L2_approx} at every noise level $\sigma>0$.
\end{lemma}

\begin{proof}[Proof sketch]
If $\sigma>t_j$, the discarded correction has magnitude at most
$2a_j/\sigma^2$. Otherwise, the oracle is exact in coordinate $j$ whenever
all preceding gates equal one, and its error is at most $2a_j/\sigma^2$ on
the remaining event. Under the noisy target, these gate failures have
Gaussian tail probabilities. Geometric scale separation controls the sum
of the discarded corrections' squared errors and makes the gate-failure
contribution small enough as well. Here we retain the noise dependence of
the tails: as $\sigma$ decreases, they decay fast enough to compensate for
the $\sigma^{-4}$ factor in the squared correction. Summing the two
contributions gives \eqref{eqn:L2_approx} uniformly over $\sigma>0$. \Cref{app:box-accuracy} gives the
complete calculation.
\end{proof}

\paragraph{One new coordinate per round.} The disjoint gate neighborhoods allow us to simulate a sampler while exposing at most one additional label in each round. 

\begin{lemma}[One new coordinate per round]
\label{lem:box-simulation}
Let a sampler use $R<d$ rounds and at most $P$  queries per round. Its actual output $Y$ can be coupled to an
output $\widetilde Y$ depending only on the sampler's random seed and
$\theta_1,\ldots,\theta_R$, such that
\[
  \Pr[Y\ne\widetilde Y]\le\frac{RP}{M},
\]
where the probability is over the independent labels and the sampler’s seed.
\end{lemma}

The complete statement and proof appear in \Cref{app:box-simulation}. With these lemmas in hand, we are ready to prove \Cref{thm:box-parallel}.

\begin{proof}[Proof of \Cref{thm:box-parallel}]
The construction and \Cref{lem:box-accuracy} give the required targets and
oracles. Fix a sampler using $R<d$ rounds and couple its output as in
\Cref{lem:box-simulation}. Conditional on the first $R$ labels and the seed,
$\widetilde Y_{R+1}$ is fixed and $\theta_{R+1}$ remains uniform. Since the
candidate intervals are disjoint,
\[
  \Pr[\widetilde Y\in K^\theta]
  \le \Pr[\widetilde Y_{R+1}\in I_{R+1}(\theta_{R+1})]
  \le\frac1M.
\]
The coupling therefore gives
\[
  \E_\theta [Q^\theta(K^\theta)]
  =\Pr[Y\in K^\theta]
  \le\frac{RP+1}{M}.
\]
The choice of $M$ in \Cref{app:box-parameters} makes this at most $0.01$ for every $R<d$.
By Markov's inequality, $Q^\theta(K^\theta)\le0.1$ with probability at least
$0.9$ over $\theta$. For each such instance, $\pi^\theta(K^\theta)=1$, and hence
\[
  \TV(Q^\theta,\pi^\theta)
  \ge\pi^\theta(K^\theta)-Q^\theta(K^\theta)\ge0.9.
\]
This is the claimed probability bound.
\end{proof}

\paragraph{Matching upper bound.} Finally, we prove a matching upper bound. In particular, we show that for the same parameter choice, the selected intervals can also be recovered one coordinate at a time.

\begin{lemma}[Upper bound for the same target family]
\label{lem:box-upper}
For sufficiently large $d$, every target in the constructed family can be
sampled to total variation error at most $1/10$ using $O(d)$ rounds and
$O(\log d)$ queries per round, for every fixed legal oracle.
\end{lemma}

Tweedie’s formula expresses the posterior mean in terms of the score, and this mean lies in the selected interval. The separation of coordinate scales lets us choose query distributions under which the \(L^2\) score guarantee provides sufficiently accurate estimates of this mean. A constant number of refinements per coordinate, using \(O(\log d)\) queries each, identifies all selected intervals with sufficiently high probability. The algorithm then draws a fresh uniform sample from the recovered box. \Cref{app:box-recovery} gives the algorithm and proof.


%% file: conclusion.tex
We establish polynomial lower bounds on the number of sequential
rounds needed for diffusion sampling with approximate scores.
For a fixed polynomial query budget per round, we prove a
$\widetilde{\Omega}(d^{1/3})$ lower bound for smooth, near-isotropic
Gaussian mixtures and an $\Omega(d)$ lower bound for anisotropic
boxes. The latter matches an $O(d)$-round algorithm for our box family.

These results show that mean-square score accuracy alone does not guarantee polylogarithmic-round sampling, even for smooth, near-isotropic targets. Our constructions use fixed approximate oracles that remain accurate under the corresponding noisy target distribution at every noise level while making informative queries depend on information learned in earlier rounds.

Our results suggest several directions for future work.
\begin{enumerate}
    \item \textbf{Tightening the near-isotropic bounds.}
    For our near-isotropic Gaussian mixture family, a polynomial
    gap remains between the $\widetilde{\Omega}(d^{1/3})$ lower
    bound and the $\widetilde{O}(d^{5/6})$ upper bound obtained by
    specializing the analysis of \cite{ChenEtAl2026HighAccuracy}.
    Since this upper bound uses a sequential algorithm, exploiting
    parallel queries may reduce the number of rounds.
    On the lower bound side, can a sharper analysis or a modified
    construction improve the exponent?

    \item \textbf{Near-isotropic convex bodies.}
    Our box construction gives a linear lower bound for uniform
    sampling from convex bodies, but its proof relies on strong
    anisotropy. Can polynomial round lower bounds also be
    established for uniform sampling from near-isotropic convex
    bodies using approximate score oracles?
    This would show that convexity and near-isotropy together
    do not suffice for polylogarithmic-round parallel sampling in this model.

    \item \textbf{Relating sequential and parallel lower bounds.}
    The $\widetilde{\Omega}(d^{1/2})$ sequential query lower bound
    of \cite{XP26} raises the question of whether a comparable
    parallel round lower bound holds for near-isotropic targets.
    More generally, when can techniques for sequential query
    lower bounds be adapted to withstand polynomially many
    queries per round?
    The counting-oracle setting provides another context for
    studying this relationship \cite{AGR24}.

    \item \textbf{Stronger score guarantees.} Do polynomial parallel-round lower bounds persist under exact score access or uniform pointwise score-error guarantees? Our constructions exploit the freedom to suppress information on atypical query points while maintaining accuracy under the noisy target distribution. Stronger guarantees may require a different mechanism for enforcing sequential dependence.

\end{enumerate}

%% file: appendix_simplified.tex
\providecommand{\poly}{\operatorname{poly}}
\providecommand{\niTV}{d_{\mathrm{TV}}}
\providecommand{\niKL}{D_{\mathrm{KL}}}
\providecommand{\niMI}{\mathrm{MI}}
\providecommand{\niUnif}{\operatorname{Unif}}
\providecommand{\niLip}{\operatorname{Lip}}
\providecommand{\niCov}{\operatorname{Cov}}
\providecommand{\niScoreError}{\varepsilon_{\mathrm{score}}}
\providecommand{\niXP}{\textnormal{[XP26]}}
\providecommand{\niCCDR}{\textnormal{[CCDR26]}}
\let\niEquationLabel\label

\appendix\label{ni:appendix}
\crefalias{section}{appendix}
\crefalias{subsection}{subappendix}

\newpage
\section{Deferred proofs from \Cref{sec:near_isotropic_LB}}\label{appendix_for_near_isotropic}

We now provide details of omitted proofs from \Cref{sec:near_isotropic_LB}.

\paragraph{Shared notation and instance law.}
As in \Cref{sec:near_isotropic_LB}, $d$ is the ambient dimension, $B$ is the
number of hidden blocks, and $m=\lfloor d/B\rfloor$ is their common
dimension. Set $d_{\mathrm{pad}}=d-Bm$. The padding coordinates are independent
uniform signs; when $d_{\mathrm{pad}}=0$, their law is the unit mass on
the zero-dimensional space. For a nonempty codebook $S_j\subseteq\{-1,+1\}^m$, write
$n_j=|S_j|$ and $\kappa_j=(\log n_j)/m$.
For the full instance $\bfS=(S_1,\ldots,S_B)$ and any $\tau>0$, define
\begin{equation}\label{ni:center-law}
 \mu^{\bfS}=\left(\bigotimes_{j=1}^B\niUnif(S_j)\right)\otimes\niUnif(\{-1,+1\}^{d_{\mathrm{pad}}}),
 \qquad
 \pi^{\bfS}_\tau=\mu^{\bfS}*\mathcal N(0,\tau^2I_d).
\end{equation}
Fix a target smoothing parameter $\gamma>0$, and set
\begin{equation}\label{ni:target}
 \pi^{\bfS}=\pi^{\bfS}_\gamma,\qquad s^{\bfS}_\tau=\nabla\log\pi^{\bfS}_\tau,\qquad
 (\pi^{\bfS})^{(\sigma)}=\pi^{\bfS}_{\sqrt{\gamma^2+\sigma^2}}.
\end{equation}
Thus $\sigma>0$ is the additional query noise, while
$\tau=\sqrt{\gamma^2+\sigma^2}$ is total noise. The target itself is at
$\tau=\gamma$. When a random instance is used, its codebooks are
independent, and $S_j$ is uniform among all $n_j$-element subsets of the
cube. All instance probabilities refer to this unconditioned product law
unless conditioning is explicitly indicated.  We use $S$ for a single codebook and
$\mathbf{S}=(S_1,\ldots,S_B)$ for the full collection of codebooks. All logarithms are natural.

\paragraph{Public parameters.}
Unless stated otherwise, fix an integer query budget
$P=P(d)=\operatorname{poly}(d)\ge 1$, a score accuracy
$0<\varepsilon_{\mathrm{score}}
=\operatorname{poly}(d)^{-1}\le 1$ tending to zero, and a tolerance
$\operatorname{poly}(d)^{-1}\le \rho=\rho(d)\le 1/10$.
These polynomial bounds are fixed before the instance family is
constructed. All public parameter choices are independent of the
algorithm and sampled instance. Implicit constants may depend on
$\gamma$, these fixed polynomial bounds, and the fixed rate interval
chosen below.

\subsection{Parameters and near-isotropy}
\label{ni:parameters}
In this section, we choose the block dimensions and codebook sizes and verify
near-isotropy.

As in \Cref{sec:near_isotropic_LB}, let
$\alpha_\gamma=\operatorname{MI}(U;U+\gamma Z)$, for independent
$U\sim\operatorname{Unif}(\{-1,+1\})$ and $Z\sim\mathcal N(0,1)$.
Fix
\[
 0<\kappa_{\min}<\kappa_{\max}<\alpha_\gamma,
 \qquad c_0=\frac{\alpha_\gamma-\kappa_{\max}}2.
\]
For any proposed dimensions $1\le B\le d$ and
$\max\{2,\sqrt2\,\gamma\}\le m\le d$, set
\begin{equation}\label{ni:global-budgets}
 \delta=\frac{\rho^2}{100Pd^3},
 \qquad \eta=\frac{\niScoreError}{8B},
 \qquad
 \zeta(\tau)=\min\left\{\frac{\rho^2}{100},
                   \frac{\eta^2\tau^2}{16m}\right\}
 \quad(\tau\ge\gamma).
\end{equation}
These are the hiding, local score-error, and gate-failure budgets.
Choose
\begin{equation}\label{eqn:def_H}
    H=\Theta_\gamma\!\left(
        \log\frac{dP}{\varepsilon_{\mathrm{score}}\rho}
    \right),
\end{equation}
with a sufficiently large implicit constant, uniformly over the
proposed dimensions $B,m\le d$.
For the fixed polynomial budgets, $H=\Theta(\log d)$.

\begin{lemma}[Public parameter choice]
\label{ni:parameter-lemma}\label{ni:packing}
For all sufficiently large $d$, there are public integers
$B,n_1,\ldots,n_B$ such that, with $m=\lfloor d/B\rfloor$,
\begin{equation}\label{ni:dimension-scales}
 B=\Theta\!\left((d/\log d)^{1/3}\right),\qquad
 m=\Theta\!\left(d^{2/3}(\log d)^{1/3}\right),\qquad
 H=\Theta(\log d).
\end{equation}
Each $n_j$ is a power of two, the rates $\kappa_j=(\log n_j)/m$
belong to $[\kappa_{\min},\kappa_{\max}]$, and
\begin{equation}\label{ni:packed-gaps}
 \kappa_{j+1}-\kappa_j\ge
 \frac{\kappa_{\max}-\kappa_{\min}}{2B}\qquad(j<B).
\end{equation}
The constant factor in $B$ can be chosen so that these gaps exceed any
prescribed constant multiple of $\sqrt{H/m}+H/m$.
The product has public cardinality
$N_d=2^{d_{\mathrm{pad}}}\prod_jn_j=\exp(\Theta(d))$.
Furthermore, $m\ge\max\{2,\sqrt2\,\gamma\}$ and
\begin{equation}\label{ni:parameter-test-margin}
 \begin{gathered}
 \frac1\gamma\sqrt{\frac{8\log(1/\zeta(\gamma))}{m}}\le c_0,
 \qquad
 \zeta(\gamma),\ e^{-c_0m}\le\frac{\rho^2}{100},\qquad
 RPB\delta\le\frac{\rho^2}{100d}\quad(0\le R<B).
 \end{gathered}
\end{equation}
\end{lemma}

\begin{proof}
Choose $B$ to be a sufficiently small constant multiple of
$(d/\log d)^{1/3}$, rounded down, and set $m=\lfloor d/B\rfloor$.
This gives \eqref{ni:dimension-scales}. Define
\[
 n_j=2^{\left\lceil\frac m{\log2}
       \left(\kappa_{\min}
          +\frac{j-1}{B}(\kappa_{\max}-\kappa_{\min})\right)
                      \right\rceil}.
\]

Rounding changes each prescribed rate by less than $\log2/m$.
Consequently, for sufficiently large $d$, all rates lie in the fixed
interval and
\[
 \kappa_{j+1}-\kappa_j\ge
 \frac{\kappa_{\max}-\kappa_{\min}}B-\frac{\log2}m
 \ge\frac{\kappa_{\max}-\kappa_{\min}}{2B}.
\]
Since $H=O(\log d)$ and $m=\Theta(d/B)$,
\[
 B\left(\sqrt{\frac Hm}+\frac Hm\right)
 =O\!\left(\sqrt{\frac{B^3\log d}{d}}
                    +\frac{B^2\log d}{d}\right).
\]
The first term is as small as needed by the choice of the constant in
$B$, while the second and $B/m$ tend to zero. This proves
\eqref{ni:packed-gaps} and the claimed separation. In particular, the gaps can be made larger than the uniform
transition width in \Cref{ni:window}.
Moreover,
\[
\log N_d=d_{\mathrm{pad}}\log2+m\sum_{j=1}^B\kappa_j
          =\Theta(d),
\]
because $Bm=d-O(B)$ and the rates stay in a fixed positive interval.

Finally, $H/m\to0$ proves the first margin in
\eqref{ni:parameter-test-margin}, and $m/\log d\to\infty$ makes
$e^{-c_0m}$ smaller than the inverse-polynomial quantity $\rho^2/100$.
The gate-failure bound is immediate from its definition. Since
$R,B\le d$, the hiding budget gives
$RPB\delta\le d^2P\delta=\rho^2/(100d)$.
\end{proof}

\paragraph{Instance geometry.}
Each target $\pi^{\mathbf S}$ is an equal-weight mixture of $N_d$
Gaussians with distinct centers and common component covariance
$\gamma^2 I_d$.

\begin{lemma}[Near-isotropy]\label{ni:near-isotropy}
For the public parameters above, define
\begin{equation}\label{ni:moment-event}
 G=\left\{\bfS:
 \|\mathbb E_{\pi^{\bfS}}X\|_2\le d^{-10},\quad
 \|\operatorname{Cov}(\pi^{\bfS})-(1+\gamma^2)I_d\|_{\mathrm{op}}
                        \le d^{-10}\right\}.
\end{equation}
Then $\Pr_{\bfS}(G^c)\le\exp(-\Omega(m))$.
More generally, for arbitrary prescribed codebook sizes and every
$t\in(0,1]$, the probability that either norm in
\eqref{ni:moment-event} exceeds $t$ is at most
$4m^2t^{-2}\sum_jn_j^{-1}$.
\end{lemma}

\begin{proof}
Write
\[
 v_j=\frac1{n_j}\sum_{y\in S_j}y,\qquad
 M_j=\frac1{n_j}\sum_{y\in S_j}yy^\top,\qquad
 v=(v_1,\ldots,v_B,0_{d_{\mathrm{pad}}}).
\]
For any mean-zero, variance-one function $f$ on the cube, sampling
without replacement gives
\[
 \mathbb E_S\left(\frac1n\sum_{y\in S}f(y)\right)^2
       =\frac{2^m-n}{n(2^m-1)}\le\frac1n.
\]
Indeed, for a uniform ordered distinct pair,
$\mathbb E f(Y)f(Y')=-1/(2^m-1)$. Apply this identity to $y_a$
and $y_ay_b$ for $a\ne b$. The diagonal entries of $M_j$ are one, so
\[
Z_{\bfS}:=\|v\|_2^2+\sum_j\|M_j-I_m\|_F^2,
 \qquad \mathbb E_{\bfS}Z_{\bfS}\le m^2\sum_j\frac1{n_j}.
\]
For every fixed instance,
\[
 \mathbb E_{\pi^{\bfS}}X=v,\qquad
 \operatorname{Cov}(\pi^{\bfS})
 =\operatorname{diag}(M_1-v_1v_1^\top,\ldots,
       M_B-v_Bv_B^\top,I_{d_{\mathrm{pad}}})+\gamma^2I_d.
\]
If $Z_{\bfS}\le t^2/4$, then the mean has norm at most $t/2$ and
\[
 \|\operatorname{Cov}(\pi^{\bfS})-(1+\gamma^2)I_d\|_{\mathrm{op}}
 \le\sqrt{Z_{\bfS}}+Z_{\bfS}\le t/2+t^2/4\le t.
\]
Markov's inequality therefore gives
\[
\Pr_{\bfS}[\text{either norm exceeds }t]
 \le\frac{4m^2}{t^2}\sum_j\frac1{n_j}.
\]
For the chosen sizes, $n_j\ge e^{\kappa_{\min}m}$ and $Bm\le d$.
Taking $t=d^{-10}$ yields
$\Pr_{\bfS}(G^c)\le4d^{21}m e^{-\kappa_{\min}m}=\exp(-\Omega(m))$.
\end{proof}



\subsection{Technical inputs from \cite{XP26}}
\label{ni:analytic}
In this subsection, we state variants of technical inputs from \cite{XP26} we need for the proof.

For the one-block estimates, let $m\ge1$, $V_m=\{-1,+1\}^m$, and
$U_m=\niUnif(V_m)$. For every nonempty $S\subseteq V_m$ and noise $\tau>0$, set
\[
 p^S_\tau=\niUnif(S)*\mathcal N(0,\tau^2I_m),\qquad
 q_\tau=U_m*\mathcal N(0,\tau^2I_m),\qquad
 s^S_\tau=\nabla\log p^S_\tau,\qquad s^0_\tau=\nabla\log q_\tau.
\]
Thus $s^0_\tau$ is the null score from \Cref{sec:near_isotropic_LB}.
Let $\varphi_\tau$ be the density of $\mathcal N(0,\tau^2I_m)$, and let
$Y\sim U_m$ and $Z\sim\mathcal N(0,I_m)$ be independent. Define
\[
 L_\tau(y,x)=\frac{\varphi_\tau(x-y)}{q_\tau(x)},\qquad
 \ell_\tau(y,x)=\log L_\tau(y,x),\qquad
 \alpha_\tau=\frac1m\mathbb E\ell_\tau(Y,Y+\tau Z).
\]
By the product structure, $\alpha_\tau$ equals
$\niMI(Y^0;Y^0+\tau Z^0)$ for independent
$Y^0\sim\niUnif(\{-1,+1\})$ and $Z^0\sim\mathcal N(0,1)$, exactly as in
\eqref{eq:alpha-definition}. We write
$\|f\|_{\infty,2} :=\sup_x\|f(x)\|_2$ and $\niLip(f)$ for the global
Euclidean Lipschitz constant. The mixture formula gives
\begin{equation}\label{ni:likelihood-identities}
 \mathbb E_{Y\sim U_m}L_\tau(Y,x)=1,\qquad L_\tau(y,x)\le2^m,\qquad
 \frac{p^S_\tau(x)}{q_\tau(x)}=\frac1{|S|}\sum_{y\in S}L_\tau(y,x).
\end{equation}
The numbered results below refer to arXiv:2604.10857v1 of
\cite{XP26}. We use its Gaussian-smoothing identities and KL--Fisher comparison. The without-replacement adaptation is proved here.

\begin{lemma}[Posterior score bounds; {\cite{XP26}, Proposition B.6}]
\label{ni:posterior}
Let $\nu$ be supported on $[-1,1]^m$, let $Y\sim\nu$, and let
$Z\sim\mathcal N(0,I_m)$ be independent. For every $\tau>0$, define
$s_{\nu,\tau}=\nabla\log(\nu*\mathcal N(0,\tau^2I_m))$. Then
\begin{align}
 s_{\nu,\tau}(x)&=\frac{\mathbb E[Y\mid Y+\tau Z=x]-x}{\tau^2},
 \label{ni:posterior-score}\\
 \nabla s_{\nu,\tau}(x)&=-\frac{I_m}{\tau^2}
       +\frac{\niCov(Y\mid Y+\tau Z=x)}{\tau^4}.
 \label{ni:posterior-hessian}
\end{align}
Consequently,
\begin{align}
 \niLip(s_{\nu,\tau})&\le\frac1{\tau^2}+\frac m{\tau^4},\qquad
 \|s^S_\tau-s^0_\tau\|_{\infty,2}\le\frac{2\sqrt m}{\tau^2},
 \label{ni:score-bounds}\\
 \|\nabla_x\ell_\tau(y,x)\|_2&\le\frac{2\sqrt m}{\tau^2}
       \qquad(y\in V_m).
 \label{ni:likelihood-gradient}
\end{align}
\end{lemma}
\begin{proof}[Justification of the specialization]
The identities and score bounds are the specialization of \niXP,
Proposition B.6(a)--(d), with support-radius parameter $1$.
Posterior means lie in $[-1,1]^m$, and the posterior covariance has operator
norm at most $m$. These facts give \eqref{ni:score-bounds}. Finally,
\[
 \nabla_x\ell_\tau(y,x)=\tau^{-2}
 \bigl(y-\mathbb E_{Y\sim U_m}[Y\mid Y+\tau Z=x]\bigr),
\]
which proves \eqref{ni:likelihood-gradient}.
\end{proof}

\begin{lemma}[Information-density concentration and nearby noise]
\label{ni:concentration}
For every $\tau>0$, the law of $\ell_\tau(y,y+\tau Z)$ is the same for all
$y\in V_m$, and its mean is $m\alpha_\tau$. For either sign and every $t\ge0$,
\begin{equation}\label{ni:information-tail}
 \Pr_Z\!\left[\pm\bigl(\ell_\tau(y,y+\tau Z)-m\alpha_\tau\bigr)\ge t\right]
 \le\exp\!\left(-\frac{\tau^2t^2}{8m}\right).
\end{equation}
For $0<u\le v$,
\begin{equation}\label{ni:nearby-noise}
 0\le\alpha_u-\alpha_v\le\log(v/u).
\end{equation}
\end{lemma}
\begin{proof}
Sign symmetry of the cube and Gaussian noise makes the law independent of
$y$. Equation~\eqref{ni:likelihood-gradient} makes
$z\mapsto\ell_\tau(y,y+\tau z)$ a $2\sqrt m/\tau$-Lipschitz function.
Gaussian concentration gives \eqref{ni:information-tail}. This is the
Gaussian-Lipschitz argument used in \cite{XP26}, Proposition B.8(b); the stronger
noise-uniform estimates there are not needed.

For \eqref{ni:nearby-noise}, couple $X_u=Y+uZ_1$ and
$X_v=X_u+\sqrt{v^2-u^2}Z_2$, with $Y\sim U_m$ and independent standard
Gaussians $Z_1,Z_2$. Data processing gives $\alpha_u\ge\alpha_v$.
Writing $h$ for differential entropy, independence gives, for $u<v$,
$h(X_v)\ge h(X_v\mid Z_2)=h(X_u)$. Equality handles $u=v$. Thus
\[
 \alpha_u-\alpha_v
 =\frac{h(X_u)-h(X_v)+h(vZ_1)-h(uZ_1)}m\le\log(v/u).
\]
\end{proof}

\paragraph{KL--Fisher comparison.}
For $\gamma>0$ and $m\ge\max\{2,\sqrt2\,\gamma\}$, define
\[
 \widetilde\tau=\tau\sqrt{1-\frac{\gamma^2}{m^2}},\qquad
 F_\tau(S)=\mathbb E_{X\sim p^S_\tau}
              \|s^S_\tau(X)-s^0_\tau(X)\|_2^2.
\]
Here $F_\tau(S)$ is the squared error of the null score on the hidden
smoothed distribution. These definitions apply at every positive noise level.

\begin{lemma}[KL--Fisher comparison; {\cite{XP26}, Proposition B.7(b)}]
\label{ni:kl-fisher}
For every nonempty $S\subseteq V_m$ and every $\tau\ge\gamma$,
\begin{equation}\label{ni:fisher-comparison}
 F_\tau(S)\le\frac{2em^2}{\gamma^2\tau^2}
       \niKL(p^S_{\widetilde\tau}\|q_{\widetilde\tau}).
\end{equation}
\end{lemma}
\begin{proof}
Set $\lambda=1-\gamma^2/m^2\ge 1/2$, so that
$\widetilde\tau=\sqrt{\lambda}\,\tau$.
Since the full-cube reference is a product of smoothed uniform signs,
\[
    \nabla^2\log q_{\sqrt t}(x)\preceq t^{-2}I_m
    \qquad (t>0).
\]
The comparison in \cite[Proposition~B.7(b)]{XP26} therefore gives
\[
\begin{aligned}
F_\tau(S)
&\le
\frac{2}{(1-\lambda)\tau^2}
\exp\!\left(
    \frac{2(1-\lambda)}{\lambda\tau^2}
\right)
D_{\mathrm{KL}}\!\left(
    p_{\widetilde\tau}^{S}\,\middle\|\,q_{\widetilde\tau}
\right) \\
&\le
\frac{2em^2}{\gamma^2\tau^2}
D_{\mathrm{KL}}\!\left(
    p_{\widetilde\tau}^{S}\,\middle\|\,q_{\widetilde\tau}
\right),
\end{aligned}
\]
where the second inequality uses $\tau\ge\gamma$ and
\[
    \frac{2(1-\lambda)}{\lambda\tau^2}
    \le \frac{4}{m^2}\le 1.
\]
The finite Gaussian mixtures satisfy the comparison's hypotheses.
The analytic estimates used at $\widetilde\tau$ hold for every
positive noise level, including when $\widetilde\tau<\gamma$.
\end{proof}

\begin{lemma}[Without-replacement adaptation of {\cite{XP26}, Lemma B.10}]
\label{ni:without-replacement}
Let $S$ be a uniformly random $n$-element subset of $V_m$, with
$1\le n\le2^m$. For every $\tau>0$ and $u\ge0$,
\begin{equation}\label{ni:kl-average}
 \begin{aligned}
 \mathbb E_S\niKL(p^S_\tau\|q_\tau)
 &\le\mathbb E_{Y\sim U_m,Z}
       \log\!\left(1+\frac{L_\tau(Y,Y+\tau Z)}n\right)\\
 &\le e^{-u}+\log(1+2^m)
       \Pr[\ell_\tau(Y,Y+\tau Z)\ge\log n-u].
 \end{aligned}
\end{equation}
\end{lemma}
\begin{proof}
The source lemma samples codewords independently. For the required
without-replacement version, put $M=2^m$ and give $S$ a uniform random
ordering $(Y_1,\ldots,Y_n)$. Let $X=Y_1+\tau Z$, and put
$L_i=L_\tau(Y_i,X)$. Exchangeability and
\eqref{ni:likelihood-identities} give
\[
 \mathbb E_S\niKL(p^S_\tau\|q_\tau)
 =\mathbb E\log\!\left(\frac1n\sum_{i=1}^nL_i\right).
\]
Conditional on $(Y_1,X)$, each $Y_i$, $i\ge2$, is uniform on the other
$M-1$ vertices. Since $\sum_yL_\tau(y,X)=M$,
\[
 \mathbb E[L_i\mid Y_1,X]=\frac{M-L_1}{M-1}.
\]
Conditional Jensen therefore yields
\[
 \mathbb E_S\niKL(p^S_\tau\|q_\tau)
 \le\mathbb E\log\!\left(
       \frac{L_1+(n-1)(M-L_1)/(M-1)}n\right)
 \le\mathbb E\log(1+L_1/n).
\]
For the last inequality, discard the nonpositive term
$-(n-1)L_1/(M-1)$ and use $(n-1)M/(M-1)\le n$.
The marginal law of $(Y_1,X)$ is the planted law in the statement.
Since $L_1\le M$, the logarithm is at most $e^{-u}$ on
$\{\log L_1<\log n-u\}$ and at most $\log(1+M)$ on its complement.
This proves the second inequality. No independence among the remaining
codewords is used.
\end{proof}

\subsection{Proof of \Cref{eq:prefix-oracle}}
\label{ni:accuracy-section}\label{ni:local}
\begin{lemma}[Complete version of \Cref{eq:prefix-oracle}]
\label{ni:global-accuracy}
Fix an integer $1\le B\le d$ for which $m=\lfloor d/B\rfloor\ge\max\{2,\sqrt2\,\gamma\}$.
The codebooks $S_1,\ldots,S_B$ may be any fixed nonempty subsets of
$\{-1,+1\}^m$. Use $\eta$ and $\zeta$ from \eqref{ni:global-budgets}.
For every nonempty one-block codebook $S$, use the quantities of
\Cref{ni:analytic} to define
\begin{equation}\label{ni:cutoff}
 t_\tau=m\alpha_\tau-\frac1\tau\sqrt{8m\log\frac1{\zeta(\tau)}}.
\end{equation}
\begin{equation}\label{ni:gate-definition}
 a^S_\tau(x)=\min\left\{1,\max\left\{0,
             1+\max_{y\in S}\ell_\tau(y,x)-t_\tau\right\}\right\}.
\end{equation}
\begin{equation}\label{ni:local-response}
 b^S_\tau(x)=
 \begin{cases}
 s^0_\tau(x),&F_\tau(S)\le\eta^2/(4\tau^2),\\
 s^S_\tau(x),&F_\tau(S)>\eta^2/(4\tau^2).
 \end{cases}
\end{equation}
For the product instance, define, for $j\in[B]$,
\begin{equation}\label{ni:full-oracle}
 \widehat s^{\bfS}_{j,\tau}(x)
 =s^0_\tau(x_j)+\left(\prod_{i<j}a^{S_i}_\tau(x_i)\right)
                    \bigl(b^{S_j}_\tau(x_j)-s^0_\tau(x_j)\bigr),
\end{equation}
and use the exact score on the public padding coordinates. The empty
product is one. The resulting field is fixed once $\bfS$ is chosen, and satisfies, for every $\tau\ge\gamma$,
\begin{equation}\label{ni:global-error}
 \mathbb E_{X\sim\pi^{\bfS}_\tau}
       \|\widehat s^{\bfS}_\tau(X)-s^{\bfS}_\tau(X)\|_2^2
 \le\frac{B(B+1)\eta^2}{8\tau^2}
 \le\frac{\niScoreError^2}{256\tau^2}.
\end{equation}
Consequently $\widehat s^{\bfS}(x,\sigma)
=\widehat s^{\bfS}_{\sqrt{\gamma^2+\sigma^2}}(x)$ is a legal oracle for $\pi^{\bfS}$
at score accuracy $\niScoreError$.
\end{lemma}

The distinction between the two branches of \eqref{ni:local-response} is
important. If a block uses the null score, its error is already small.
Otherwise it uses the exact local score, and its only error comes from
an earlier gate failing. We first establish the two bounds needed for
this calculation. At a fixed total noise, the main text abbreviates
$s_j=s^{S_j}_\tau$, $s^0=s^0_\tau$, $a_j=a^{S_j}_\tau$, and
$b_j=b^{S_j}_\tau$.

\begin{lemma}[Fixed-instance local estimates]
\label{ni:local-accuracy}
For every nonempty one-block codebook $S$ and every $\tau\ge\gamma$,
\begin{align}
 p^S_\tau[a^S_\tau=1]&\ge1-\zeta(\tau),
 \label{ni:gate-completeness}\\
 \mathbb E_{p^S_\tau}\|b^S_\tau-s^S_\tau\|_2^2
     &\le\frac{\eta^2}{4\tau^2},\qquad
 \|b^S_\tau-s^0_\tau\|_{\infty,2}\le\frac{2\sqrt m}{\tau^2}.
 \label{ni:local-size}
\end{align}
\end{lemma}
\begin{proof}
Generate $X=Y+\tau Z$, with $Y\sim\niUnif(S)$ and $Z\sim\mathcal N(0,I_m)$
independent. Conditional on $Y=y$, \Cref{ni:concentration} and the
choice of $t_\tau$ imply
\[
 \Pr_Z[\ell_\tau(y,y+\tau Z)<t_\tau]
 \le\exp\!\left(-\frac{\tau^2}{8m}
                 \frac{8m\log(1/\zeta(\tau))}{\tau^2}\right)
 =\zeta(\tau).
\]
Since $y\in S$, on the complementary event the maximum in
\eqref{ni:gate-definition} is at least $t_\tau$ and the gate equals one.
Averaging over the generating codeword proves
\eqref{ni:gate-completeness}. This argument is for every fixed $S$;
no randomness of the codebook is required.

If the first branch of \eqref{ni:local-response} is selected, its squared
error is $F_\tau(S)\le\eta^2/(4\tau^2)$. In the second branch its error
is zero. Also $b^S_\tau-s^0_\tau$ is either zero or $s^S_\tau-s^0_\tau$,
so \Cref{ni:posterior} gives the uniform size bound.

\end{proof}

\begin{proof}[Proof of \Cref{ni:global-accuracy}]

\Cref{ni:local-accuracy} supplies the acceptance and local-error bounds used in the proof of \Cref{eq:prefix-oracle}. Applying that calculation gives
\[
\mathbb E\bigl\|\widehat s_\tau^{\mathbf S}(X)-s_\tau^{\mathbf S}(X)\bigr\|_2^2
\le
\frac{B\eta^2}{4\tau^2}
+\frac{2mB(B-1)\zeta(\tau)}{\tau^4}
\le
\frac{B(B+1)\eta^2}{8\tau^2}.
\]
Substituting \(\eta=\varepsilon_{\mathrm{score}}/(8B)\) proves \eqref{ni:global-error}, and \(\tau^2=\gamma^2+\sigma^2\ge \sigma^2\) gives the required additional-noise guarantee.
\end{proof}

\paragraph{Lipschitzness of the oracle.}
For completeness, the explicit oracle is also spatially Lipschitz at
each fixed noise level. This is not an assumption of our oracle model.
The finite maximum and clipping in the gate preserve the bound in
\eqref{ni:likelihood-gradient}, so
\begin{equation}\label{ni:local-lipschitz}
 \niLip(a^S_\tau)\le\frac{2\sqrt m}{\tau^2},\qquad
 \niLip(b^S_\tau)\le\frac1{\tau^2}+\frac m{\tau^4}.
\end{equation}
To see the full-dimensional bound, put $r_i=\|x_i-y_i\|_2$ and telescope
the product of prefix gates. For factors in $[0,1]$,
$|\prod_i u_i-\prod_i v_i|\le\sum_i|u_i-v_i|$. At a fixed prefix-gate
value the local field is a convex combination of $s^0_\tau$ and
$b^{S_j}_\tau$, so
\[
 \|\widehat s^{\bfS}_{j,\tau}(x)-\widehat s^{\bfS}_{j,\tau}(y)\|_2
 \le(\tau^{-2}+m\tau^{-4})r_j+4m\tau^{-4}\sum_{i<j}r_i.
\]
The strictly lower-triangular matrix of ones has operator norm at most
$B-1$. Since $Bm\le d$, including the exact padding gives
$\niLip(\widehat s^{\bfS}_\tau)\le\tau^{-2}+4d\tau^{-4}$. Since \(\tau\ge\gamma\), this is \(O_\gamma(d)\) uniformly over the query noise level.

\subsection{Proof of \Cref{lm:one-block-per-round}}
\label{ni:simulation-section}\label{ni:adaptive}
\begin{lemma}[Complete version of \Cref{lm:one-block-per-round}]
\label{ni:codebook-simulation}
Use the dimensions and oracle of \Cref{ni:global-accuracy}, the budgets in \eqref{ni:global-budgets},
and $H$ from \eqref{eqn:def_H}. Under the independent codebook law specified
above, suppose the adjacent rates are separated by a sufficiently
large constant multiple of
\begin{equation}\label{ni:local-envelope}\niEquationLabel{ni:simulation-gap-condition}
 \sqrt{\frac Hm}+\frac Hm.
\end{equation}
The required constant depends only on $\gamma$.
For any randomized algorithm using at most
$P$ queries per round and $0\le R<B$ rounds, let $\mathsf T_R,Y$ denote
its actual transcript and output. There are coupled simulated variables
$\widetilde{\mathsf T}_R,\widetilde Y$ satisfying
\begin{equation}\label{ni:simulation-error}
 \Pr[\mathsf T_R\ne\widetilde{\mathsf T}_R]\le RPB\delta,
 \qquad \Pr[Y\ne\widetilde Y]\le RPB\delta,
\end{equation}
such that the simulated transcript and output are functions
only of the algorithm's independent random seed and $S_1,\ldots,S_R$.
Here the probabilities average over the independent codebooks and seed.
\end{lemma}

The proof has two distinct parts. The codebook estimates show that one
batch can be answered without any codebook after the first unknown one.
The progressive coupling then turns this fixed-batch statement into an
adaptive statement. 

\paragraph{The simulated oracle.}
For $r\in[B]$, define
\begin{equation}\label{ni:truncation}
 \widehat s^{\bfS,[r]}_{j,\tau}(x)=
 \begin{cases}
 \widehat s^{\bfS}_{j,\tau}(x),&j\le r,\\
 s^0_\tau(x_j),&j>r,
 \end{cases}
\end{equation}
retaining the exact padding response. This truncation depends only on
$S_1,\ldots,S_r$. It is used for analysis and need not be a legal oracle.
The progressive interaction uses this $r$-truncated response in round $r$.
All algorithmic and output randomness is included in a seed $\omega$
independent of $\bfS$. Early stopping is padded with empty batches. A
bounded-round algorithm promised only good instances can be extended
arbitrarily to other transcripts with the same bound.

\paragraph{Sparse hiding and dense null responses.}
For $\tau\ge\gamma$, define the public thresholds
\begin{equation}\label{ni:sparse-threshold}
 \kappa^-(\tau)=\frac{t_\tau-\log(1/\delta)-1}{m},
\end{equation}
\begin{equation}\label{ni:dense-threshold}
 \kappa^+(\tau)=\alpha_{\widetilde\tau}+\frac{H}m
              +\frac1{\widetilde\tau}\sqrt{\frac{8H}m},
 \qquad \widetilde\tau=\sqrt{1-\gamma^2/m^2}\,\tau.
\end{equation}
These thresholds do not depend on $n$ or on the realization of a
codebook. By our choice of $H$, we have $H\ge 1$.
\begin{lemma}[One-block hiding estimates]
\label{ni:hiding}
For each fixed $\tau\ge\gamma$, let $S$ be a uniformly random
$n$-element subset of $V_m$, and write $\kappa=(\log n)/m$. Then
\begin{align}
 \kappa\le\kappa^-(\tau)&\ \Longrightarrow\
     \sup_x\Pr_S[a^S_\tau(x)>0]\le\delta,
 \label{ni:sparse-hiding}\\
 \kappa\ge\kappa^+(\tau)&\ \Longrightarrow\
     \Pr_S[b^S_\tau\equiv s^0_\tau]\ge1-\delta.
 \label{ni:dense-hiding}
\end{align}
The second event means equality as functions of $x$, not just equality
at a particular query point.
\end{lemma}
\begin{proof}
For sparse hiding, fix $x$. A positive gate requires some $y\in S$
with $\ell_\tau(y,x)>t_\tau-1$. For a uniform cube vertex $Y$,
\eqref{ni:likelihood-identities} and Markov's inequality imply
\[
 \Pr_Y[\ell_\tau(Y,x)>t_\tau-1]
 \le e^{1-t_\tau}\mathbb E_YL_\tau(Y,x)=e^{1-t_\tau}.
\]
Give $S$ a uniform random ordering. Every codeword marginal is uniform
on the cube, so a union bound, without independence of those codewords,
gives
\begin{equation}\label{ni:blind-gate}
 \Pr_S[a^S_\tau(x)>0]\le ne^{1-t_\tau}
 =\exp(m\kappa+1-t_\tau)\le\delta.
\end{equation}
This bound holds for every fixed $x$.

For dense nullity, the rate assumption in
\eqref{ni:dense-threshold} gives
\[
 \log n-H-m\alpha_{\widetilde\tau}
       \ge\frac{\sqrt{8Hm}}{\widetilde\tau}.
\]
Apply \Cref{ni:without-replacement} at noise $\widetilde\tau$
with $u=H$. The tail in that lemma is at most $e^{-H}$ by
\Cref{ni:concentration}. Because $\log(1+2^m)\ge1$,
\[
 \begin{aligned}
 \mathbb E_S\niKL(p^S_{\widetilde\tau}\|q_{\widetilde\tau})
 &\le(1+\log(1+2^m))e^{-H}\\
 &\le2\log(1+2^m)e^{-H}
 \le\frac{\delta\gamma^2\eta^2}{8em^2}.
 \end{aligned}
\]
The final inequality follows from the choice of the sufficiently
large implicit constant in \eqref{eqn:def_H}. \Cref{ni:kl-fisher} now yields
\[
 \mathbb E_SF_\tau(S)
 \le\frac{2em^2}{\gamma^2\tau^2}
       \frac{\delta\gamma^2\eta^2}{8em^2}
 =\frac{\delta\eta^2}{4\tau^2}.
\]
Markov's inequality gives
$\Pr_S[F_\tau(S)>\eta^2/(4\tau^2)]\le\delta$. On the complement,
\eqref{ni:local-response} selects the null score at every point.
\end{proof}

\begin{lemma}[Uniform transition width]\label{ni:window}
For every $\tau\ge\gamma$,
\begin{equation}\label{ni:width}
 \kappa^-(\tau)\le\alpha_\tau\le\kappa^+(\tau),\qquad
 \kappa^+(\tau)-\kappa^-(\tau)
 =O_\gamma\!\left(\sqrt{\frac Hm}+\frac Hm\right).
\end{equation}
\end{lemma}
\begin{proof}
The first inequalities follow from $t_\tau\le m\alpha_\tau$ and
$\alpha_{\widetilde\tau}\ge\alpha_\tau$.
Since $\gamma^2/m^2\le1/2$, the nearby-noise bound gives
\[
 \alpha_{\widetilde\tau}-\alpha_\tau
 \le-\frac12\log(1-\gamma^2/m^2)\le\frac{\gamma^2}{m^2}.
\]
Using $\widetilde\tau\ge\gamma/\sqrt2$ and
$\zeta(\tau)\ge\zeta(\gamma)$, subtract the thresholds to obtain
\[
 \begin{aligned}
 \kappa^+(\tau)-\kappa^-(\tau)
 &\le\frac{\gamma^2}{m^2}
       +\frac{H+\log(1/\delta)+1}{m}\\
 &\quad+\frac{\sqrt2}{\gamma}\sqrt{\frac{8H}{m}}
       +\frac1\gamma\sqrt{\frac{8\log(1/\zeta(\gamma))}{m}}.
 \end{aligned}
\]
Our choice of $H$ guarantees $H\ge 1$,
$H\ge \log(1/\delta)$, and
$H\ge \log(1/\zeta(\gamma))$.
\end{proof}

\begin{lemma}[One-batch indistinguishability]
\label{ni:one-batch}
Suppose the rate separation in \Cref{ni:codebook-simulation} holds.
Fix $r\in[B]$, an exposed
prefix $S_1,\ldots,S_{r-1}$, and a seed $\omega$. Any batch of at most
$P$ queries determined by this prefix and seed receives different
responses from the full and $r$-truncated oracles with probability at
most $PB\delta$, over the independent suffix $S_r,\ldots,S_B$.
\end{lemma}
\begin{proof}
For $r=B$, the two oracles are identical. Otherwise fix a query
$(x,\tau)$ from the batch. For a block $j>r$, the response difference is
\[
 \left(\prod_{i<j}a^{S_i}_\tau(x_i)\right)
                 \bigl(b^{S_j}_\tau(x_j)-s^0_\tau(x_j)\bigr).
\]
If $\kappa_r\le\kappa^-(\tau)$, any nonzero difference requires
$a^{S_r}_\tau(x_r)>0$, an event of probability at most $\delta$ by
\Cref{ni:hiding}. This single event controls the entire later
suffix. If instead $\kappa_r>\kappa^-(\tau)$, then for every $j>r$,
\[
 \kappa_j\ge\kappa_{r+1}
 >\kappa_r+\bigl(\kappa^+(\tau)-\kappa^-(\tau)\bigr)
 >\kappa^+(\tau),
\]
where the rate separation dominates the width by \Cref{ni:window}. Thus
$b^{S_j}_\tau\equiv s^0_\tau$ except on an event of probability at
most $\delta$ for each such $j$. In either case, a union bound bounds
the probability of a different answer to this query by $(B-r)\delta$.
Another union bound over the batch gives
$P(B-r)\delta\le PB\delta$.

The conditioning fixes both the locations and the noise levels of the
batch, but leaves the suffix codebooks independent with their original
laws. Hence the fixed-noise estimates above apply separately to each
query. No event simultaneous over a continuum of noise levels, and no
noise discretization, is needed.
\end{proof}

\paragraph{The general progressive coupling.}
The remaining argument uses no properties of scores or codebooks. We
state it separately so that the box proof can invoke exactly the same
lemma.
\begin{lemma}[Progressive simulation]
\label{ni:progressive}
Let $n,P\ge1$ be integers, let $\Theta=(\Theta_1,\ldots,\Theta_n)$ have
independent finite-valued coordinates, and let $\omega$ be an independent
seed containing all algorithmic and output randomness. Fix an
oracle $O_\Theta$, truncated oracles $O^{[r]}_\Theta$ for
$r\in[n]$, and $\delta_{\mathrm{batch}}\ge0$. Assume that for every $r\in[n]$:
\begin{enumerate}
\item $O^{[r]}_\Theta$ depends only on $(\Theta_1,\ldots,\Theta_r)$;
\item for every fixed prefix $(\Theta_1,\ldots,\Theta_{r-1})$ in its
support, seed, and batch of at most $P$ queries determined by them, the
full and truncated oracles disagree on the batch with probability at most
$\delta_{\mathrm{batch}}$ over the independent suffix.
\end{enumerate}
For every integer $0\le R\le n$ and every algorithm using at most $R$
rounds and $P$ queries per round, couple its actual interaction to the
progressive interaction that answers round $r$ with $O^{[r]}_\Theta$.
Their transcripts and outputs disagree with probability at most
$R\delta_{\mathrm{batch}}$. The progressive transcript and output depend
only on $(\omega,\Theta_1,\ldots,\Theta_R)$.
\end{lemma}
\begin{proof}
Run the progressive interaction with $O^{[r]}_\Theta$ in round $r$.
Its transcript after that round depends only on
$(\omega,\Theta_1,\ldots,\Theta_r)$, by induction. Thus its round-$r$
batch is independent of the fresh suffix conditional on the prefix and seed.
Let $D_r$ be the event that the full and truncated oracles disagree on
this progressive batch. The hypothesis, applied for each prefix and seed
and then averaged, gives $\Pr(D_r)\le\delta_{\mathrm{batch}}$.
Couple the actual interaction using the same $(\Theta,\omega)$.
If the transcripts first disagree in round $r$, they chose the same batch
in that round, so this first-disagreement event is contained in $D_r$.
Their disagreement probability is at most
$\sum_{r=1}^R\Pr(D_r)\le R\delta_{\mathrm{batch}}$. The same seed determines
the outputs, proving the output bound. For $R=0$, the outputs can be
identical. No estimate is conditioned on earlier agreement. Independence
is used only in the progressive interaction, where prefix dependence
holds identically.
\end{proof}

\begin{proof}[Proof of \Cref{ni:codebook-simulation}]
For any exposed prefix and seed, \Cref{ni:one-batch} bounds the probability of a different
answer to the next batch by $PB\delta$. Apply \Cref{ni:progressive} with
$\Theta_j=S_j$, $n=B$, $O_\Theta=\widehat s^{\bfS}$,
$O^{[r]}_\Theta=\widehat s^{\bfS,[r]}$, and
$\delta_{\mathrm{batch}}=PB\delta$. It gives
\eqref{ni:simulation-error} and the required prefix dependence.
No conditioning on near-isotropy is used here. The public parameters in
\Cref{ni:parameter-lemma} give $RPB\delta\le\rho^2/(100d)$;
with $\rho=1/10$, this is smaller than $0.001$, as required by
\Cref{lm:one-block-per-round}.
\end{proof}

\subsection{Proof of \Cref{lm:hidden}}
\label{ni:test-section}
\begin{lemma}[Complete version of \Cref{lm:hidden}]
\label{ni:target-test}
Use the one-block gate, cutoff, and budgets of
\Cref{ni:global-accuracy}. Use $\kappa_{\max}$ and $c_0$ fixed in \Cref{ni:parameters}. Suppose
\begin{equation}\label{ni:test-condition}
 \frac1\gamma\sqrt{\frac{8\log(1/\zeta(\gamma))}m}\le c_0.
\end{equation}
For an integer $1\le n\le2^m$ with $(\log n)/m\le\kappa_{\max}$, and
for each $n$-element codebook $S\subseteq V_m$, define the target-level
test
\begin{equation}\label{ni:gate-test}
 E(S)=\{x\in\mathbb R^m:a^S_\gamma(x)=1\}
 =\left\{x\in\mathbb R^m:
                   \max_{y\in S}\ell_\gamma(y,x)\ge t_\gamma\right\}.
\end{equation}
Every fixed codebook satisfies
\begin{equation}\label{ni:test-completeness}
 p^S_\gamma(E(S))\ge1-\zeta(\gamma).
\end{equation}
If $S$ is uniform among the $n$-element subsets of $V_m$, then
\begin{equation}\label{ni:test-bounds}
 \sup_{x\in\mathbb R^m}\Pr_S[x\in E(S)]
     \le ne^{-t_\gamma}\le e^{-c_0m}.
\end{equation}
In particular, for any random $V\in\mathbb R^m$ independent of this
random codebook, $\Pr[V\in E(S)]\le e^{-c_0m}$.
\end{lemma}
\begin{proof}
The test is evaluated at total noise \(\gamma\), corresponding to the target distribution before additional query noise is added. Equation~\eqref{ni:test-completeness}
is exactly the gate-completeness estimate
\eqref{ni:gate-completeness} at this noise level.

For blindness, fix $x\in\mathbb R^m$. A uniform cube vertex $Y$ obeys
\[
 \Pr_Y[\ell_\gamma(Y,x)\ge t_\gamma]
 =\Pr_Y[L_\gamma(Y,x)\ge e^{t_\gamma}]
 \le e^{-t_\gamma}\mathbb E_YL_\gamma(Y,x)=e^{-t_\gamma}.
\]
A union bound over a uniform random ordering of $S$ gives
$\Pr_S[x\in E(S)]\le ne^{-t_\gamma}$. The margins in
\eqref{ni:test-condition} imply
\[
 t_\gamma\ge m(\alpha_\gamma-c_0)
             =m(\kappa_{\max}+c_0),
\]
so $n\le e^{m\kappa_{\max}}$ yields the second inequality in
\eqref{ni:test-bounds}. Finally, condition on the value of $V$ and
average the fixed-point bound. Independence ensures that conditioning on
$V$ does not change the uniform codebook law.
\end{proof}

For the public sizes chosen in \Cref{ni:parameter-lemma}, the two
failure bounds above are at most $\rho^2/100$. Taking $\rho=1/10$
therefore gives the constants of \Cref{lm:hidden}, with slack.
The restriction to the prescribed rates is necessary: this is not a
claim about arbitrary subsets of the full cube.

Applying the unconditioned TV calculation in the proof of \Cref{thm:parallel-lower-bound} with the bounds of \Cref{ni:codebook-simulation} and \Cref{ni:target-test} gives
\[
\mathbb E_{\mathbf S}d_{\mathrm{TV}}(Q_{\mathbf S},\pi^{\mathbf S})
\ge
1-\zeta(\gamma)-e^{-c_0m}-RPB\delta,
\]
where the expectation is under the independent codebook law.

\subsection{Proof of \Cref{lem:section33-upper-bound}}
\label{ni:upper}
\begin{lemma}[Complete version of \Cref{lem:section33-upper-bound}]
\label{ni:upper-bound}
Fix $\gamma>0$. For every sufficiently large $d$, there is a randomized
sampler that, for every product target $\pi^{\bfS}$ with the parameters of
\Cref{ni:parameter-lemma} and every fixed legal oracle with
accuracy $\niScoreError$, uses at most
\[
 O_\gamma\bigl(\sqrt{dm}\,(\log d)^3\bigr)
\]
queries and has an output law $Q$ satisfying
\begin{equation}\label{ni:upper-error}
 \niTV(Q,\pi^{\bfS})
 =O_\gamma\!\bigl(d^{-1/2}+\niScoreError\sqrt{\log d}\bigr).
\end{equation}
The query bound is a worst-case bound, and the algorithm may use one query
per round. In particular, accuracy
$\niScoreError=\poly(d)^{-1}$ suffices for total variation error $1/10$ using
$\widetilde O(d^{5/6})$ queries and adaptive rounds.
\end{lemma}

We apply \cite{ChenEtAl2026HighAccuracy} using the following pointwise bound on the true denoiser. No regularity of the approximate oracle is assumed.

\begin{lemma}[Regularity of the true denoiser]
\label{ni:denoiser}
Let $U\sim\mu^{\bfS}$, let $Z_0,Z\sim\mathcal N(0,I_d)$ be independent of
each other and of $U$, and put $X_0=U+\gamma Z_0$. For $\sigma>0$, define
$D_\sigma(x)=\mathbb E[X_0\mid X_0+\sigma Z=x]$ using the Gaussian
posterior at every $x$. Then
\begin{equation}\label{ni:denoiser-bound}
 \sup_{\sigma>0,\ x\in\mathbb R^d}\|\nabla D_\sigma(x)\|_F
       \le\sqrt d+\frac{\sqrt{dm}}{4\gamma^2}.
\end{equation}
\end{lemma}
\begin{proof}
Write $\tau^2=\gamma^2+\sigma^2$ and
$\Sigma_\sigma(x)=\niCov(U\mid X_0+\sigma Z=x)$.
Gaussian conditioning and differentiation of the posterior mean give
\begin{equation}\label{ni:denoiser-identity}
 \nabla D_\sigma(x)=\frac{\gamma^2}{\tau^2}I_d
                        +\frac{\sigma^2}{\tau^4}\Sigma_\sigma(x).
\end{equation}
Indeed, conditionally on $U=u$ and $X_0+\sigma Z=x$, the mean of $X_0$ is
$(\gamma^2x+\sigma^2u)/\tau^2$, while the derivative of the posterior
mean of $U$ is $\Sigma_\sigma(x)/\tau^2$.

The posterior of $U$ retains its product structure across the public
blocks. Each hidden covariance block is positive semidefinite and has
trace at most $m$, since its coordinates take values in $\{-1,+1\}$.
The padding coordinates remain independent, with posterior variances at
most one. Consequently,
\[
 \operatorname{tr}\Sigma_\sigma(x)\le d,\qquad
 \|\Sigma_\sigma(x)\|_{\mathrm{op}}\le m,\qquad
 \|\Sigma_\sigma(x)\|_F^2\le
 \|\Sigma_\sigma(x)\|_{\mathrm{op}}\operatorname{tr}\Sigma_\sigma(x)\le dm.
\]
Now use $\gamma^2/\tau^2\le1$ and
$\sigma^2/(\gamma^2+\sigma^2)^2\le1/(4\gamma^2)$ in
\eqref{ni:denoiser-identity}.
\end{proof}

\begin{proof}[Proof of \Cref{ni:upper-bound}]
Write $\pi=\pi^{\bfS}$. The pointwise bound in \Cref{ni:denoiser} verifies Assumptions 4.6 and 4.8 of \cite{ChenEtAl2026HighAccuracy}, including the hypothesis under its
auxiliary one-step distribution. We apply \cite[Theorem 4.9]{ChenEtAl2026HighAccuracy} with data distribution
$\pi$ and additional Gaussian noise; these are precisely the scores
available in the oracle model.

Choose a step size $h$ that is a sufficiently small constant multiple
of $1/(\sqrt{dm}(\log d)^2)$, with the constant depending only on
$\gamma$, and use the schedule
\[
 \sigma_1^2=\gamma^2/d,\qquad
 \sigma_{k+1}^2=(1+h)\sigma_k^2.
\]
Let $K$ be the first index with $\sigma_K^2\ge(1+\gamma^2)d^2$.
Then
\begin{equation}\label{ni:noise-schedule}
 K=O_\gamma\!\bigl(\sqrt{dm}(\log d)^3\bigr),\qquad
 Kh=O_\gamma(\log d).
\end{equation}
With internal tolerance $d^{-10}$, this schedule meets the cited
theorem's step-size condition. If $Q_*$ is the output law before
query-budget truncation, the theorem gives
\begin{equation}\label{ni:upper-kl}
 \begin{aligned}
 \niKL(\pi^{(\sigma_1)}\|Q_*)
 =O\Bigl(&\niKL(\pi^{(\sigma_K)}\|\mathcal N(0,\sigma_K^2I_d))
              +Kd^{-10}\\
 &+h\sum_{k=1}^K\sigma_k^2
       \mathbb E_{\pi^{(\sigma_k)}}
       \|\widehat s(\cdot,\sigma_k)-s_{\pi,\sigma_k}\|_2^2\Bigr).
 \end{aligned}
\end{equation}
Here $\widehat s$ is the given legal additional-noise oracle. For $k<K$,
$h\sigma_k^2=\sigma_{k+1}^2-\sigma_k^2$; the terminal summand is included
only to simplify the error bound.

The identity $\mathbb E\|X_0\|_2^2=(1+\gamma^2)d$ holds for every
instance. Convexity of relative entropy yields
\[
 \niKL(\pi^{(\sigma_K)}\|\mathcal N(0,\sigma_K^2I_d))
       \le\frac{\mathbb E\|X_0\|_2^2}{2\sigma_K^2}\le\frac1{2d}.
\]
Oracle legality and \eqref{ni:noise-schedule} bound the score-error sum by
\[
 \niScoreError^2Kh=O_\gamma(\niScoreError^2\log d).
\]
Also $Kd^{-10}=O_\gamma(d^{-9}(\log d)^3)$. Thus
\eqref{ni:upper-kl} is at most
$O_\gamma(d^{-1}+\niScoreError^2\log d)$.

To remove the additional smoothing, note that $\pi$ and $\pi^{(\sigma_1)}$
have the same centers and weights. Relative-entropy convexity applied to
corresponding Gaussian components gives
\[
 \niKL(\pi\|\pi^{(\sigma_1)})
 \le\frac d2\left(\log(1+1/d)-\frac1{d+1}\right)\le\frac1{4d},
\]
using $\log(1+r)-r/(1+r)\le r^2/2$ for $r\ge0$.
Pinsker's inequality and the triangle inequality give
\[
 \niTV(Q_*,\pi)=O_\gamma(d^{-1/2}+\niScoreError\sqrt{\log d}).
\]
Finally, \cite{ChenEtAl2026HighAccuracy}, Theorem 3.1, bounds the total number of rejection-sampling
queries by $O(K+\log d)$ with probability at least $1-d^{-2}$.
Stop after this deterministic budget and output an arbitrary point if it
is exceeded. This changes the output law by at most $d^{-2}$ and proves
\eqref{ni:upper-error} with a worst-case query bound.
Since $m=\Theta(d^{2/3}(\log d)^{1/3})$, the bound is
$\widetilde O(d^{5/6})$. The error tends to zero for
$\niScoreError=\poly(d)^{-1}$.
\end{proof}

\section{Deferred proofs from \Cref{sec:convex_anisotropic_LB}}\label{app:boxes}
We give the public parameters and prove the three lemmas used for the box
construction. The same parameter choice supports the lower bound and the
recovery algorithm. Throughout, the ambient dimension is $d$, and $\sigma>0$
denotes additional Gaussian noise.

\subsection{Public parameters and gates}\label{app:box-parameters}
Fix a polynomial query budget $P=\poly(d)$ and an inverse-polynomial
score accuracy $\niScoreError=\poly(d)^{-1}$. These public bounds are fixed
before choosing the instance family. Implicit constants may depend on these
fixed bounds. For sufficiently large $d$, put
\begin{equation}\label{eq:appB-parameters}
\begin{gathered}
 M=\left\lceil100(dP+1)\right\rceil,
 \qquad q=\frac{100Md}{\niScoreError},
 \qquad a_i=\frac14q^{-(i-1)},\\
 h_i=\frac{a_i}{100M},\qquad
 t_i=\frac{4a_i}{\niScoreError}\quad (i\in[d]).
\end{gathered}
\end{equation}
Here $q$ is the ratio between consecutive spatial scales, $h_i$ is the
interval half-length, and $t_i$ is the noise cutoff.
For $k\in[M]$, define
\begin{equation}\label{eq:appB-intervals}
 c_i(k)=a_i+\left(k-\frac12\right)\frac{a_i}{M},
 \qquad I_i(k)=[c_i(k)-h_i,c_i(k)+h_i].
\end{equation}
Choose $\theta_1,\ldots,\theta_d$ independently and uniformly from $[M]$,
and let
\[
 K^\theta=\prod_{i=1}^d I_i(\theta_i),
 \qquad \pi^\theta=\operatorname{Unif}(K^\theta).
\]
All parameters and candidate intervals are public; only the labels are hidden.

For an interval $I$, write $\operatorname{dist}(z,I)=\inf_{u\in I}|z-u|$.
For $i<d$, set $g_{i,\sigma}^k\equiv1$ when $\sigma>t_{i+1}$, and otherwise set
\begin{equation}\label{eq:appB-gates}
 g_{i,\sigma}^k(z)=
 \min\left\{1,\max\left\{0,
  2-\frac{\operatorname{dist}(z,I_i(k))}{10h_i}\right\}\right\}.
\end{equation}
Thus, at low noise, the gate is one within distance $10h_i$ of its interval
and zero at distance at least $20h_i$.

\begin{lemma}[Geometry and gates]\label{lem:appB-geometry}
For the parameters above, every $K^\theta$ is full-dimensional and contained
in $B_2(0,1)$. The intervals $I_i(k)$ lie in $[a_i,2a_i]$, are pairwise
disjoint, and have center spacing $a_i/M=100h_i$.
Moreover, for $i<d$:
\begin{enumerate}
\item If $\sigma\le t_{i+1}$, the sets
      $\{z:g_{i,\sigma}^k(z)>0\}$ are pairwise disjoint over $k\in[M]$.
\item For independent $U_i\sim\operatorname{Unif}(I_i(k))$ and
      $Z_i\sim N(0,1)$, at every $\sigma>0$,
\begin{equation}\label{eq:appB-gate-tail}
 \Pr\bigl[g_{i,\sigma}^k(U_i+\sigma Z_i)<1\bigr]
 \le 2\exp\left(-\frac{50h_i^2}{\sigma^2}\right).
\end{equation}
\end{enumerate}
\end{lemma}
\begin{proof}
Since $h_i<a_i/(2M)$, the interval endpoints lie strictly between $a_i$
and $2a_i$, and distinct intervals are disjoint. Their positive lengths give
full dimension. Since $q\ge2$, every $x\in K^\theta$ satisfies
\[
 \|x\|_2^2\le4\sum_{i=1}^d a_i^2
 \le\frac{1}{4(1-q^{-2})}\le\frac13<1.
\]
At low noise, a positive gate requires
$|z-c_i(k)|<21h_i<50h_i$,
which proves disjointness. If $|\sigma Z_i|\le10h_i$, then
$\operatorname{dist}(U_i+\sigma Z_i,I_i(k))\le10h_i$, so the gate equals one.
The Gaussian tail bound proves \eqref{eq:appB-gate-tail} for $\sigma\le t_{i+1}$;
for larger $\sigma$, the gate is identically one.
\end{proof}

\subsection{Proof of \Cref{lem:box-accuracy}}\label{app:box-accuracy}
For $k\in[M]$, let $m_{i,\sigma}^k(z)$ be the posterior mean of
$U\sim\operatorname{Unif}(I_i(k))$ given $U+\sigma Z=z$, where
$Z\sim N(0,1)$ is independent. Equivalently, this is the mean obtained by
weighting the uniform interval measure by its Gaussian likelihood at $z$.
It is defined at every $z$ and belongs to
$I_i(k)\subset[a_i,2a_i]$.
The true score is
\[
 s_{i,\sigma}^\theta(x)
 =-\frac{x_i}{\sigma^2}
   +\frac{m_{i,\sigma}^{\theta_i}(x_i)}{\sigma^2}.
\]
Using the gates in \eqref{eq:appB-gates}, define the fixed oracle
\begin{equation}\label{eq:appB-oracle}
 \widehat s_{j,\sigma}^\theta(x)
 =-\frac{x_j}{\sigma^2}
  +\mathbf 1_{\{\sigma\le t_j\}}
    \left(\prod_{i<j}g_{i,\sigma}^{\theta_i}(x_i)\right)
    \frac{m_{j,\sigma}^{\theta_j}(x_j)}{\sigma^2}.
\end{equation}
The empty product is one. This is the oracle in \cref{eq:box-gated-oracle} of the main text.

\begin{lemma}[Complete version of \Cref{lem:box-accuracy}]\label{lem:appB-accuracy}
For every instance $\theta$, the oracle \eqref{eq:appB-oracle} is fixed before interaction and satisfies, for every $\sigma>0$,
\begin{equation}\label{eq:appB-accuracy}
 \mathbb E_{X\sim(\pi^\theta)^{(\sigma)}}
 \|\widehat s_\sigma^\theta(X)-s_\sigma^\theta(X)\|_2^2
 \le\frac{\niScoreError^2}{2\sigma^2}.
\end{equation}
In particular, it is legal at score accuracy $\niScoreError$.
\end{lemma}
\begin{proof}
Fix $\theta$ and $\sigma$. Call coordinate $j$ active when $\sigma\le t_j$.
In an inactive coordinate, the squared error is at most $4a_j^2/\sigma^4$.
In an active coordinate, the error vanishes if every preceding gate is one;
otherwise its absolute value is at most $2a_j/\sigma^2$.
\Cref{lem:appB-geometry} and a union bound therefore give
\begin{equation}\label{eq:appB-error-split}
 \mathbb E\|\widehat s_\sigma^\theta(X)-s_\sigma^\theta(X)\|_2^2
 \le
 \frac4{\sigma^4}\sum_{j:\,t_j<\sigma}a_j^2
 +\frac8{\sigma^4}\sum_{j:\,\sigma\le t_j}a_j^2
                    \sum_{i<j}e^{-50h_i^2/\sigma^2}.
\end{equation}

For the first term, if an inactive coordinate exists, let $j_0$ be the first
one. Since $t_{j_0}=4a_{j_0}/\niScoreError<\sigma$, geometric summation yields
\begin{equation}\label{eq:appB-inactive-error}
 \frac4{\sigma^4}\sum_{j\ge j_0}a_j^2
 \le\frac{4a_{j_0}^2}{(1-q^{-2})\sigma^4}
 \le\frac{\niScoreError^2}{4(1-q^{-2})\sigma^2}
 \le\frac{\niScoreError^2}{3\sigma^2}.
\end{equation}
If all coordinates are active, this term is zero.

For the gate term, interchange the sums and extend the sum over $j$:
\[
 \frac8{\sigma^4}\sum_{j:\,\sigma\le t_j}a_j^2
                    \sum_{i<j}e^{-50h_i^2/\sigma^2}
 \le\frac8{(1-q^{-2})\sigma^4}
       \sum_{i<d}a_{i+1}^2e^{-50h_i^2/\sigma^2}.
\]
For every $u>0$, $ue^{-u}\le e^{-1}$. Applying this with
$u=50h_i^2/\sigma^2$ bounds the last display by
\begin{align*}
 \frac{8}{50e(1-q^{-2})\sigma^2}
       \sum_{i<d}\left(\frac{a_{i+1}}{h_i}\right)^2
 &=\frac{4(d-1)\niScoreError^2}
         {25e(1-q^{-2})d^2\sigma^2}
 \le\frac{\niScoreError^2}{10\sigma^2}.
\end{align*}
Here $a_{i+1}/h_i=100M/q=\niScoreError/d$;
the final bound follows from $q\ge2$ and $(d-1)/d^2\le1/4$.
Together with \eqref{eq:appB-inactive-error}, this proves
\eqref{eq:appB-accuracy}. The calculation retains the dependence on $\sigma$
and hence applies also as $\sigma$ tends to zero.
\end{proof}

\subsection{Proof of \Cref{lem:box-simulation}}\label{app:box-simulation}
The simulation exposes at most one additional label in each round.
For $r\in[d]$, define the truncated oracle
\begin{equation}\label{eq:appB-truncated}
 \widehat s_{j,\sigma}^{\theta,[r]}(x)
 =\begin{cases}
   \widehat s_{j,\sigma}^{\theta}(x),&j\le r,\\
   -x_j/\sigma^2,&j>r.
  \end{cases}
\end{equation}
It depends only on $\theta_1,\ldots,\theta_r$ and need not itself be legal.
The progressive interaction answers round $r$ with
$\widehat s^{\theta,[r]}$.

\begin{lemma}[Complete version of \Cref{lem:box-simulation}]\label{lem:appB-simulation}
Let a randomized algorithm make at most $P$ queries per round for $R<d$
rounds. Couple its actual interaction with the progressive interaction using
the same instance and random seed. Their transcripts, and hence their outputs
$Y,\widetilde Y$, disagree with probability at most $RP/M$.
The progressive transcript and output depend only on
$(\theta_1,\ldots,\theta_R)$ and the seed.
The probability averages over the independent labels and seed.
\end{lemma}
\begin{proof}
We verify the hypotheses of \Cref{ni:progressive} for the full oracle
$\widehat{s}^{\theta}$ and the truncated oracles in \eqref{eq:appB-truncated}.
Each truncation depends only on $\theta_1,\ldots,\theta_r$.
For $r=d$, the two oracles coincide. Fix $r<d$, a prefix
$\theta_1,\ldots,\theta_{r-1}$, and the independent seed $\omega$
containing all algorithmic and output randomness. Consider any batch
of at most $P$ queries determined by this prefix and seed. The remaining
labels are independent and uniform.

For a query $(x,\sigma)$ in this batch, the oracles agree in coordinates
$j\le r$. If $\sigma>t_{r+1}$, every coordinate $j>r$ is inactive,
so the oracles agree there as well. Otherwise, any disagreement in a
coordinate $j>r$ requires
\[
g_{r,\sigma}^{\theta_r}(x_r)>0.
\]
By \Cref{lem:appB-geometry}, this can hold for at most one value of $\theta_r$,
so the conditional probability of disagreement on this query is at most
$1/M$. A union bound over the batch gives disagreement probability at most
$P/M$, even when its queries use different noise levels.

\Cref{ni:progressive} therefore applies with $n=d$, $\Theta_i=\theta_i$,
and $\delta_{\mathrm{batch}}=P/M$. It gives disagreement probability
at most $RP/M$ for the coupled transcripts and outputs, together with
the claimed dependence on the first $R$ labels and the seed.
\end{proof}


\subsection{Proof of Lemma 4.4}\label{app:box-recovery}
We now fix a target instance. The oracle in this subsection can be any fixed oracle satisfying the score-error guarantee. It need not
be the oracle constructed in \eqref{eq:appB-oracle}.

\begin{lemma}[Complete version of Lemma 4.4]\label{lem:appB-recovery}
For all sufficiently large $d$, the family with parameters \eqref{eq:appB-parameters}
admits a randomized algorithm which, for every instance $\theta$ and every
fixed legal oracle of accuracy $\niScoreError$, recovers all labels with
probability at least $9/10$ using $O(d)$ rounds and $O(\log d)$ queries
per round.
It consequently samples $\pi^\theta$ to total variation error at most $1/10$
within the same query bounds.
\end{lemma}

It suffices to recover one interval at a time. Once all labels are correct,
we can draw a fresh exact uniform sample from their Cartesian product.
The two facts below show how a batch of approximate-score queries refines
the center estimate for the current interval.

\begin{lemma}[Gaussian change of measure]\label{lem:appB-gaussian-change}
Let $U$ be a random vector in $\mathbb R^n$, let $\sigma>0$, and let $p$
be the density of $U+\sigma Z$ for independent $Z\sim N(0,I_n)$.
For $v\in\mathbb R^n$, let $f$ be the density of $N(v,\sigma^2 I_n)$.
Then
\begin{equation}\label{eq:appB-change-measure}
 \int_{\mathbb R^n}\frac{f(x)^2}{p(x)}\,dx
 \le\mathbb E\exp\left(\frac{\|U-v\|_2^2}{\sigma^2}\right).
\end{equation}
\end{lemma}
\begin{proof}
Write $\phi_\sigma$ for the centered Gaussian density, so
$p(x)=\mathbb E_U\phi_\sigma(x-U)$ and $f(x)=\phi_\sigma(x-v)$.
Convexity of $z\mapsto1/z$, followed by Tonelli's theorem, gives
\[
 \int\frac{f(x)^2}{p(x)}\,dx
 \le\mathbb E_U\int
       \frac{\phi_\sigma(x-v)^2}{\phi_\sigma(x-U)}\,dx
 =\mathbb E_U\exp\left(\frac{\|U-v\|_2^2}{\sigma^2}\right).
\]
The equality follows by completing the square.
\end{proof}

\begin{samepage}
\begin{lemma}[One refinement]\label{lem:appB-refinement}
Suppose the labels $\theta_1,\ldots,\theta_{i-1}$ have been recovered,
and let $b\in\mathbb R$ and $r\ge h_i$ satisfy
$|b-c_i(\theta_i)|\le r$. Set $\sigma=4r$.
For each of an odd number $N$ of independent queries, draw
\begin{equation}\label{eq:appB-query-proposal}
 X_j=\begin{cases}
 U_j+\sigma Z_j,
   \quad U_j\sim\operatorname{Unif}(I_j(\theta_j)),&j<i,\\
 b+\sigma Z_j,&j=i,\\
 \sigma Z_j,&j>i,
 \end{cases}
\end{equation}
with all draws independent and every $Z_j\sim N(0,1)$.
Query the oracle at $(X,\sigma)$, form
$V=X_i+\sigma^2\widehat s_{i,\sigma}(X)$, and let $b'$ be the median of
the $N$ values of $V$. Then
\begin{equation}\label{eq:appB-refinement}
 \Pr\bigl[|b'-c_i(\theta_i)|>h_i+64\niScoreError r\bigr]
 \le2^{-N/2}.
\end{equation}
\end{lemma}
\end{samepage}
\begin{proof}
For sufficiently large $d$, $q\ge400M$. Thus every $U\in K^\theta$ satisfies
\begin{equation}\label{eq:appB-tail-size}
 \sum_{j>i}U_j^2
 \le4\sum_{j>i}a_j^2
 \le\frac{4a_i^2}{q^2-1}\le h_i^2.
\end{equation}
Let $f$ be the query density in \eqref{eq:appB-query-proposal}, and let
$p$ be the density of the true noisy target $(\pi^\theta)^{(\sigma)}$.
Their prefix marginals agree exactly. On the suffix, the query distribution
is the Gaussian centered at
$v=(b,0,\ldots,0)$, while the true distribution is the Gaussian smoothing
of the unknown intervals. Every suffix point $U$ obeys
\[
 \|U-v\|_2^2
 \le(r+h_i)^2+h_i^2\le5r^2.
\]
The matching prefix factors cancel, so \Cref{lem:appB-gaussian-change}
with $\sigma=4r$ gives
\begin{equation}\label{eq:appB-proposal-moment}
 \int\frac{f(x)^2}{p(x)}\,dx\le e^{5/16}.
\end{equation}

Cauchy--Schwarz, legality, and \eqref{eq:appB-proposal-moment} give
\[
 \mathbb E_{X\sim f}
   |\widehat s_{i,\sigma}(X)-s_{i,\sigma}^\theta(X)|
 \le e^{5/32}\frac{\niScoreError}{\sigma}.
\]
By Markov's inequality, the score error exceeds $16\niScoreError/\sigma$
with query probability at most $e^{5/32}/16<1/8$.
At every query point, the exact posterior mean
$x_i+\sigma^2s_{i,\sigma}^\theta(x)$ belongs to $I_i(\theta_i)$.
Since $\sigma=4r$, it follows that
\[
 \Pr\bigl[|V-c_i(\theta_i)|>h_i+64\niScoreError r\bigr]<1/8.
\]
Because the oracle is fixed, the independently generated queries give
independent values of $V$. If their median is outside the displayed interval,
at least $(N+1)/2$ of them are outside. A union bound over subsets gives
probability at most
$2^N(1/8)^{N/2}=2^{-N/2}$, proving \eqref{eq:appB-refinement}.
\end{proof}

\begin{proof}[Proof of \Cref{lem:appB-recovery}]
For sufficiently large $d$, $64\niScoreError\le1/4$. Choose an integer
$L=O(1)$ and an odd integer $N=O(\log d)$ such that
\begin{equation}\label{eq:appB-recovery-budget}
 (64\niScoreError)^L100M\le1,
 \qquad N\ge2\log_2(10dL).
\end{equation}
Such choices exist because $M$ is polynomial and $\niScoreError$ is
inverse-polynomial in $d$.
Process coordinates $i=1,\ldots,d$ in order. Initialize
$b_0=3a_i/2$ and $r_0=a_i$. In refinement $\ell\in[L]$, apply
\Cref{lem:appB-refinement} with $b=b_{\ell-1}$ and $r=r_{\ell-1}$,
and set the deterministic radius
\begin{equation}\label{eq:appB-radius}
 r_\ell=h_i+64\niScoreError r_{\ell-1}.
\end{equation}
After these $L$ refinements, decode the label whose public center is nearest
to $b_L$, breaking any ties by the smallest label.
This specifies an algorithm on every possible history.

The initial estimate satisfies $|b_0-c_i(\theta_i)|\le a_i/2\le r_0$,
and every prescribed radius is at least $h_i$. If all refinements so far
succeed, their error guarantees give $|b_\ell-c_i(\theta_i)|\le r_\ell$.
Moreover, \eqref{eq:appB-recovery-budget} implies
$(64\niScoreError)^L a_i\le a_i/(100M)=h_i$, so
\[
 r_L=(64\niScoreError)^L a_i+h_i\sum_{j=0}^{L-1}(64\niScoreError)^j
 \le h_i+\frac{h_i}{1-64\niScoreError} <3h_i.
\]
Since distinct candidate centers are $100h_i$ apart,
successful refinements identify the correct label.

Fix the target and oracle. Conditional on any history before the first
failed refinement, all previously decoded labels are correct, the current
center estimate lies within its prescribed radius, and the next batch uses
fresh independent randomness. \Cref{lem:appB-refinement} applies uniformly
to every such history. The probability of at least one failed refinement is
therefore at most
\[
 dL\,2^{-N/2}\le1/10.
\]
On the complementary event every label is decoded correctly.

Finally, use fresh independent randomness to sample uniformly from the
decoded box. Couple this output to a uniform sample from the true box using
the same final uniform coordinates. They agree whenever all decoded labels
are correct, so their total variation distance is at most $1/10$.
Each refinement uses one round and $N=O(\log d)$ queries, for a total
of $dL=O(d)$ rounds.
\end{proof}